\documentclass[11pt,a4paper]{scrartcl}

\usepackage{amsmath,amssymb,amsfonts,amsthm}
\usepackage{mathrsfs}
\usepackage[round]{natbib}
\usepackage{graphicx}
\usepackage{paralist}%for enumerate
\usepackage{placeins} %gibt Objekte direkt aus
\usepackage{color}

\usepackage{dsfont}
\usepackage{abbreviationsGeneral-ext}

\newtheorem{theorem}{Theorem}[section]
\newtheorem{lemma}{Lemma}[section]

\newtheorem{definition}{Definition}[section]
\newtheorem{assumption}{Assumption}[section]

\DeclareMathOperator{\p}{P}
\DeclareMathOperator{\md}{md}

\newcommand{\ds}{\displaystyle}

\begin{document}

\title{Tests for scale changes based on pairwise differences}
\author{Carina Gerstenberger, Daniel Vogel, and Martin Wendler}
\date{November 13, 2016}

%\color{blue}

\maketitle

\begin{abstract}
\footnotesize
In many applications it is important to know whether the amount of fluctuation in a series of observations changes over time. In this article, we investigate different tests for detecting change in the scale of mean-stationary time series. The classical approach based on the CUSUM test applied to the squared centered, is very vulnerable to outliers and impractical for heavy-tailed data, which leads us to contemplate test statistics based on alternative, less outlier-sensitive scale estimators.

It turns out that the tests based on Gini's mean difference (the average of all pairwise distances) or generalized $Q_n$ estimators (sample quantiles of all pairwise distances) are very suitable candidates. They improve upon the classical test not only under heavy tails or in the presence of outliers, but also under normality. An explanation for this counterintuitive result is that the corresponding long-run variance estimates are less affected by a scale change than in the case of the sample-variance-based test.

We use recent results on the process convergence of U-statistics and U-quantiles for dependent sequences to derive the limiting distribution of the test statistics and propose estimators for the long-run variance. We perform a simulations study to investigate the finite sample behavior of the test and their power. Furthermore, we demonstrate the applicability of the new change-point detection methods at two real-life data examples from hydrology and finance.
\end{abstract}

%\begin{keywords} 

%\end{keywords} 

%/////////////////////////////////////////////////////////////////////////////////////////////////////////////////////
%/////////////////////
%////////////
%/////////
%///////
%\\\\\\\
%\\\\\\\\\
%\\\\\\\\\\\\
%\\\\\\\\\\\\\\\\\\\\\
%\\\\\\\\\\\\\\\\\\\\\\\\\\\\\\\\\\\\\\\\\\\\\\\\\\\\\\\\\\\\\\\\\\\\\\\\\\\\\\\\\\\\\\\\\\\\\\\\\\\\\\\\\\\\\\\\\\\\

\section{Introduction}

The established approach to testing for changes in the scale of a univariate time series $X_1, \ldots, X_n$ is a CUSUM test applied to the squares of the centered observations, which may be written as 
\be \label{eq:CUSUM1}
 \hat{T}_{\sigma^2} =  \max_{1\leq k \leq n}\frac{k}{\sqrt{n}}\left| \hat{\sigma}^2_{1:k} - \hat{\sigma}^2_{1:n} \right|, 
%\hat{D}^{-1}_{\sigma^2}
\ee
where $\hat\sigma^2_{i:j}$ denotes the sample variance computed from the observations $X_i, \ldots, X_j$, $1 \le i < j \le n$.
To carry out the test in practice, the test statistic is usually divided by (the square root of) a suitable estimator of the corresponding long-run variance, cf.\ (\ref{eq:lrv.cusum}). This has first been considered by  \citet{Inclan_Tiao_1994}, who derive asymptotics for centered, normal, i.i.d.\ data. It has subsequently been extended by several authors to broader situations, e.g., \citet{Gombay_1996} allow the mean to be unknown and also propose a weighted version of the testing procedure, \citet{Lee_Park_2001} extend it to linear processes, and \citet{Wied_Arnold_Bissantz_Ziggel_2012} study the test for sequences that are $L_2$ NED on $\alpha$-mixing processes. A multivariate version was considered by \citet{aue:2009}.

The test statistic (\ref{eq:CUSUM1}) is prone to outliers. This has already been remarked by \citet{Inclan_Tiao_1994} and has led \citet{Lee_Park_2001} to consider a version of the test using trimmed observations. Outliers may affect the test decision in both directions: A single outlier suffices to make the test reject the null hypothesis at an otherwise stationary sequence, but more often one finds that outliers mask a change, and the test is generally very inefficient at heavy-tailed population distributions. An intuitive explanation is that, while outliers blow up the test statistic $\hat{T}_{\sigma^2}$, they do even more so blow up the long-run variance estimate, by which the test statistic is divided.\footnote{This applies in principle to any estimation method, bootstrapping or subsampling, not only to the kernel estimation method employed in the present article.}

 Writing the test statistic as in (\ref{eq:CUSUM1}) suggests this behavior may be largely attributed to the use of the sample variance as a scale estimator. 
The recognition of the extreme ``non-robustness'' of the sample variance and derived methods, in fact, stood at the beginning of the development of the  area of robust statistics as a whole \citep[e.g.][]{Box1953,Tukey_1960}. Thus, an intuitive way of constructing robust scale change-point tests is to replace the sample variance in (\ref{eq:CUSUM1}) by an alternative scale measure. 

We consider several popular scale estimators: the mean absolute deviation, the median absolute devation (MAD), the mean of all pairwise differences (Gini's mean difference) and sample quantiles of all pairwise differences. All of them allow an explicit formula and are computable in a finite number of steps. Particularly the latter two, the mean as well as sample quantiles of the pairwise differences  are promising candidates as they are almost as efficient as the standard deviation at normality and, hence, the improvement in robustness is expected to come at practically no loss in terms of power under normality. In fact, as it turns out, these tests can have a better power than the variance-based test also under normality. 
%since the corresponding long-run variance estimates are affected to a lesser extent in the presence of a change-point. 

The paper is organized as follows: In Section \ref{section_scale_estimators} we review the properties of the scale estimators and detail on our choice to particularly consider pairwise-difference-based estimators. Section \ref{section_asymptotic_results} states the test statistics and long-run variance estimators and contains asymptotic results. Section \ref{section_qn} addresses the question which sample quantile of the pairwise differences is most appropriate for the change-point problem. Section \ref{section_simulation_results} presents simulation results. Section \ref{section_data_example} illustrates the behavior of the tests at informative real-life data examples. Appendix \ref{app_supp} contains supplementary material for Section \ref{section_qn}. Proofs are deferred to the Appendix \ref{app_proofs}.

%/////////////////////////////////////////////////////////////////////////////////////////////////////////////////////
%/////////////////////
%////////////
%/////////
%///////
%\\\\\\\
%\\\\\\\\\
%\\\\\\\\\\\\
%\\\\\\\\\\\\\\\\\\\\\
%\\\\\\\\\\\\\\\\\\\\\\\\\\\\\\\\\\\\\\\\\\\\\\\\\\\\\\\\\\\\\\\\\\\\\\\\\\\\\\\\\\\\\\\\\\\\\\\\\\\\\\\\\\\\\\\\\\\\

\section{Scale Estimators}
\label{section_scale_estimators}

We use $\mathcal{L}(X)$ to denote the \emph{law}, i.e., the distribution, of any random variable $X$. We call any function $s: \mathcal{F} \to [0,\infty]$, 
where $\mathcal{F}$ is the set of all univariate distributions $F$, a \emph{scale measure} (or, analogously, a \emph{dispersion measure}) if it satisfies
\be \label{eq:scale.equiv}
	s(\mathcal{L}(a X + b)) = |a| s (\mathcal{L}(X))
	\qquad \mbox{ for all } a, b \in \R.
\ee
Although not being a scale measure itself, the variance $\sigma^2 = E(X-EX)^2$ is, in a lax use of the term, often referred to as such, since is closely related to the standard deviation $\sigma = \sqrt{\sigma^2}$, which is a proper scale measure in the above sense. A \emph{scale estimator} $s_n$ is then generally understood as the scale measure $s$ applied to the empirical distribution $F_n$ corresponding to the data set $\X_n = (X_1,\ldots, X_n)$. 
However, it is common in many situations to define the finite-sample version of the scale measure slightly different, due to various reasons. For instance, the sample variance $\sigma^2_n = \{(n-1)^{-1}\sum_{i=1}^{n}(X_i - \bar{X}_n)^2\}^{1/2}$ has the factor $(n-1)^{-1}$ due to the thus obtained unbiasedness. 

A concept that we will refer to repeatedly throughout is asymptotic relative efficiency. 
Letting $s_n$ be a scale estimator and $s$ the corresponding population value, the asymptotic variance $ASV(s_n) = ASV(s_n;F)$ of $s_n$ at the distribution $F$ is defined as the variance of the limiting normal distribution of $\sqrt{n}(s_n- s)$, when $s_n$ is evaluated at an independent sample $X_1,\ldots,X_n$ drawn from $F$. Generally, for two consistent estimators $a_n$ and $b_n$ estimating the same quantity $\theta \in \R$, i.e., $a_n \cip \theta$ and $b_n \cip \theta$, the asymptotic relative efficiency of $a_n$ with respect to $b_n$ at distribution $F$ is defined as 
\[
	ARE(a_n,b_n; F) = ASV(b_n;F)/ASV(a_n;F).
\]
In order to make two scale estimators $s_n^{(1)}$ and $s_n^{(2)}$ comparable efficiency-wise, we have to normalize them appropriately, and define their asymptotic relative efficiency at the population distribution $F$ as
\be \label{eq:are}
	ARE(s_n^{(1)},s_n^{(2)};F) \ = \ \frac{ASV(s_n^{(2)};F)}{ASV(s_n^{(1)};F)} \left\{\frac{s^{(1)}(F)}{s^{(2)}(F)}\right\}^2, 
\ee
where $s^{(1)}(F)$ and $s^{(2)}(F)$ denote the corresponding population values of the scale estimators $s_n^{(1)}$ and $s_n^{(2)}$, respectively.

In the following, we review some basic properties of four scale estimators: the mean deviation, the median absolute deviation (MAD), Gini's mean difference and the $Q_n$ scale estimator proposed by \citet{Rousseeuw_Croux_1992}.

Let $\md\left(F\right)$ denote the median of the distribution $F$, i.e, the center point of the interval $\{x\in \mathbb{R}| F\left(x-\right)\leq 1/2 \leq F\left(x\right)\}$, where $F\left(x-\right)$ denotes the left-hand limit. We define the mean deviation as $d(F) = E|X-\md(F)|$ and its empirical version as $d_n = \frac{1}{n-1}\sum_{i=1}^n|X_i - \md(F_n)|$. %Neither the use of the median as centering point nor the $(n-1)^{-1}$ factor is compulsary. We follow \citet{Gerstenberger_Vogel} here.

The question of whether to prefer the standard deviation or the mean deviation has become known as Eddington--Fisher debate.
The tentative winner was the standard deviation after Fisher (1920) showed that its asymptotic relative efficiency with respect to the mean deviation is 114\% at normality. However, \citet{Tukey_1960} pointed out that it is less efficient than the mean deviation if the normal distribution is slightly contaminated.
Thus the mean deviation appears to be a suitable candidate scale estimator for constructing less outlier-sensitive change-point tests.
\citet{Gerstenberger_Vogel} argue that, when pondering the mean deviation instead of the standard deviation for robustness reasons, it may be better to
use Gini's mean difference $g_n = \frac{2}{n(n-1)}\sum_{1\leq i<j\leq n}|X_i - X_j|$, i.e., the mean of the absolute distances of all pairs of observations.
The population version is $g\left(F\right) = E\left|X-Y\right|$, where $X, Y \sim F$ are independent. Gini's mean difference has qualitatively the same robustness under heavy-tailed distributions as the mean deviation, but retains an asymptotic relative efficiency with respect to the standard deviation of $98\%$ at the normal distribution \citep{Nair_1936}. 
%\citet{Gerstenberger_Vogel} thus suggest to view it as a consensus in the debate of standard deviation vs.\ mean deviation. 

Both estimators, mean deviation and Gini's mean difference, improve upon the variance and the standard deviation in terms of robustness, but are not robust in a modern understanding of the term. They both have an unbounded influence function %\citep{Gerstenberger_Vogel} 
and an asymptotic breakdown point of zero. Since robustness is, at least initially, our main motivation, it is natural to consider estimators that have been suggested particularly for that purpose.
A common highly robust scale estimator is the median absolute deviation (MAD), popularized by \citet{Hampel_1974}. The population value $m\left(F\right)$ is the median of the distribution of $|X-\md\left(F\right)| $ and the sample version $m_n = m_n\left(\X_n\right)$ is the median of the values $| X_i - \md(F_n)|$, $1 \leq i \leq n$. The MAD has a bounded influence function \citep[see][]{Huber_2009} and an asymptotic breakdown point of 50\%.
Its main drawback is its poor asymptotic efficiency under normality, which is only $37\%$ as compared to the standard deviation.
It is also unsuitable for change-in-scale detection due to other reasons that will be detailed in Sections \ref{section_qn} and \ref{section_simulation_results}. 

Similarly to going from the \emph{mean} absolute deviation to the \emph{median} absolute deviation, we may consider the median, or more generally any sample $\alpha$-quantile, of all pairwise differences. We call this estimator $Q_n^\alpha$ and the corresponding population scale measure $Q^\alpha$, i.e. $Q^\alpha = U^{-1}(\alpha) = \inf\{ x \,|\, U(x) \le \alpha \}$, where $U$ is the distribution function of $|X-Y|$ for $X, Y \sim F$ independent, and $U^{-1}$ the corresponding quantile function. For the precise definition of $Q_n^\alpha$, any sensible definition of the sample quantile can be employed. See, e.g., the nine different versions the R function quantile() offers. The asymptotic results we derive later are not affected, and any practical differences turn out to be negligible in the current context. So merely for simplicity, we define $Q_n^\alpha = U_n^{-1}(\alpha)$, where $U_n$ is the empirical distribution function associated with the sample $| X_i-X_j |$, $1\leq  i<j \leq n$. Thus letting $D_1, \ldots, D_{{n \choose 2}}$ be the elements of $\{ | X_i-X_j|\, | \,  1\leq  i<j \leq n \}$ in ascending order, we have $Q_n^\alpha = D_{\alpha n}$ if $\alpha n$ is integer and $Q_n^\alpha = D_{\lceil \alpha n \rceil}$ otherwise.\footnote{Note that the empirical 1/2-quantile in this sense does not generally coincide with the above definition of the sample median.}

In case of  Gini's mean difference, we observed that the transition from the average distance from the symmetry center to the average pairwise distance led to an increase in efficiency under normality. The effect is even more pronounced for the median distances, we have $ARE(Q_n^{0.5},\sigma_n,N(0,1)) = 86.3\%$. \citet{Rousseeuw_Croux_1993} propose to use the lower quartile, i.e., $\alpha = 1/4$, instead of the median. 
Specifically, they define the finite-sample version as the $\binom{\lfloor n/2 \rfloor + 1}{2}$th order statistic of the $\binom{n}{2}$ values $|X_i - X_j|$, $1 \le i < j \le n$. They call this estimator $Q_n$, and his become known under this name, which leads us to call the generalized version $Q_n^{\alpha}$.
The original $Q_n$ has an asymptotic relative efficiency with respect to the standard deviation at normality of $82\%$. \citet{Rousseeuw_Croux_1993} settle for the slightly lesser efficiency to achieve the maximal breakdown point of about 50\%.
However, this aspect is of much lesser relevance in the change-point context, quite the contrary, the very property of high-breakdown-point robustness is counterproductive for detecting change points. The original $Q_n$ is unsuited as a substitution for the sample variance in (\ref{eq:CUSUM1}), and a larger value of roughly $0.7 < \alpha < 0.9$ is much more appropriate for the problem at hand. We defer further explanations to Section \ref{section_qn}, where we discuss how to choose the $\alpha$ appropriately. 

These five scale measures, the standard deviation, the mean deviation $d_n$, Gini's mean difference $g_n$, the median absolute deviation (MAD) $m_n$, and the $\alpha$-sample quantile of all pairwise differences $Q_n^{\alpha}$, are the ones we restrict our attention to in the present article. They are summarized in Table~\ref{tab:scale_estimators} along with their sample versions.
There are of course many more potential scale estimators that satisfy the scale equivariance (\ref{eq:scale.equiv}) and more robust proposals in the literature, many of which include a data-adaptive re-weighting of the observations \citep[e.g.][Chapter 5]{Huber_2009}. 
% We defer the discussion of potential alternatives and extension to Section \ref{section_conclusion_and_outlook}.
In the present paper we explore the use of these common, easy-to-compute estimators in the change-point setting. They all admit explicit formulas, all can be computed in $O(n \log n)$ time, and the pairwise-difference estimators allow computing time savings for sequentially updated estimates (which are required in the change-point setting) -- more so than, e.g., implicitly defined estimators.
The two pairwise-difference based estimators, the average and the $\alpha$-quantile of all pairwise differences, possess promising statistical properties. We will mainly focus on these and derive their asymptotic distribution  under no change in the following section.

\begin{table} 
\renewcommand{\arraystretch}{1.5}
\caption{ \label{tab:scale_estimators}
	Scale estimators; $\md\left(F\right)$ denotes the median of the distribution $F$, $F_n$ its empirical distribution function.
}
\begin{tabular}{p{0.15\textwidth}|p{0.31\textwidth}|p{0.44\textwidth}}
\hline
 Scale \newline Estimator & Population value & Sample version\\ \hline
Standard \newline deviation 
& $\sigma\left(F\right)\! =\! \{E\left(X-E X\right)^2\}^{1/2}$ & $\displaystyle \sigma_n = \Big\{\frac{1}{n-1}\sum_{i=1}^{n}\left(X_i - \bar{X}_n\right)^2\Big\}^{1/2}$\\ \hline
Mean \newline deviation & $ d\left(F\right) = E\left|X-\md\left(F\right) \right|$ & $\displaystyle d_n= \frac{1}{n-1}\sum_{i=1}^n\left|X_i - \md(F_n )\right|$ \\ \hline
Gini's mean \newline difference & $g\left(F\right) = E\left|X-Y\right|$ & $\displaystyle g_n = \frac{2}{n(n-1)}\sum^{\phantom{0}}_{1\leq i<j\leq n}\left|X_i - X_j\right|$\\ \hline
Median \newline absolute \newline deviation 
& $m\left(F\right) = \md\left(G\right)$; where $G$ \newline is cdf of $\mathcal{L}\{\left|X-\md\left(F\right) \right|\} $ 
&  \smallskip $m_n = \md\left| X_i - \md(F_n)\right|$ \\ \hline
%$Q_n$ & $Q = G^{-1}\left(0.25\right)$; $ G = \mathcal{L}\{\left| X-Y\right|\}$ & $ Q_n =\{ \left| X_i-X_j \right| : 1\leq  i<j \leq n \}_{{\lfloor n/2 \rfloor + 1 \choose 2}}$\\ \hline
$Q_n^{\alpha}$ & 
$Q^{\alpha} = U^{-1}\left(\alpha\right)$; where $U$ \newline is cdf of $ \mathcal{L}\{\left| X-Y\right|\}$ &
$Q_n^{\alpha} = U_n^{-1}(\alpha)$; where $U_n$ is \newline empirical cdf of $| X_i-X_j |$, $1\!\leq\!  i\!<\!j\! \leq \!n$ \\
\hline
\end{tabular}

\end{table}

%/////////////////////////////////////////////////////////////////////////////////////////////////////////////////////
%/////////////////////
%////////////
%/////////
%///////
%\\\\\\\
%\\\\\\\\\
%\\\\\\\\\\\\
%\\\\\\\\\\\\\\\\\\\\\
%\\\\\\\\\\\\\\\\\\\\\\\\\\\\\\\\\\\\\\\\\\\\\\\\\\\\\\\\\\\\\\\\\\\\\\\\\\\\\\\\\\\\\\\\\\\\\\\\\\\\\\\\\\\\\\\\\\\\

\section{Test statistics, long-run variances and asymptotic results}
\label{section_asymptotic_results}

We first describe the data model employed: a very broad class of data generating processes, incorporating heavy tails and short-range dependence (Section \ref{subsec:model}). We then propose several change-point test statistics based on the scale estimators introduced in the previous section and provide estimates for their long-run variances (Section \ref{subsec:teststatistics}). We show asymptotic results for Gini's mean difference  and the $Q_n^{\alpha}$  based tests under the null hypothesis (Sections \ref{subsec:qn} and \ref{subsec:GMD}, respectively)  and discuss methods for an optimal bandwidth selection for the long-run variance estimation (Section \ref{subsec:bandwidth}). 

%/////////////////////////////////////////////////////////////////////////////////////////////////////////////////////
%/////////////////////
%////////////
%/////////
%///////

\subsection{The data model}
\label{subsec:model}
We assume the data $X_1,\ldots,X_n$ to follow the model
\begin{equation}\label{Modell_X}
X_i = \lambda_i Y_i + \mu, \qquad 1\leq i \leq n,
\end{equation}
where $Y_1,\ldots, Y_n$ are part of the stationary, median-centered sequence $(Y_i)_{i \in \Z}$.
 We want to test the hypothesis
\[
	H_0 : \lambda_1 = \ldots = \lambda_n
\]
against the alternative
\[
	H_1 : \exists \, k\in \{1,\ldots,n-1\} : \lambda_1 = \ldots = \lambda_k \neq \lambda_{k+1} = \ldots = \lambda_n.
\]
Note that this set-up is completely moment-free.
We allow the underlying process to be dependent, more precisely we assume $\left(Y_i\right)_{i\in\mathbb{Z}}$ to be near epoch dependent in probability on an absolutely regular process. Let us briefly introduce this kind of short-range dependence condition.
\begin{definition} \label{def:absreg}
\begin{enumerate}
\item Let $\mathcal{A},\mathcal{B}\subset \mathcal{F}$ be two $\sigma$-fields on the probability space $\left(\Omega,\mathcal{F},\p\right)$. We define the absolute regularity coefficient of $\mathcal{A}$ and $\mathcal{B}$ by
\[
\beta\left(\mathcal{A},\mathcal{B}\right) = E\sup_{A\in \mathcal{A}}\left| \p\left(A|\mathcal{B}\right)-\p\left(A\right)\right|.
\]
\item Let $\left(Z_n\right)_{n\in \mathbb{Z}}$ be a stationary process. Then the absolute regularity coefficients of $\left(Z_n\right)_{n\in \mathbb{Z}}$ are given by
\[
\beta_k=\sup_{n\in\N}\beta\left(\sigma\left(Z_1,\ldots,Z_n\right),\sigma\left(Z_{n+k},Z_{n+k+1},\ldots\right)\right).
\]
We say that $\left(Z_n\right)_{n\in \mathbb{Z}}$ is absolutely regular, if $\beta_k \rightarrow 0$ as $k\rightarrow \infty$.
\end{enumerate}
\end{definition}
The model class of absolutely regular processes is a common model for short-range dependence. But it does not include important classes of time series, e.g., not all linear processes. Therefore, we will not study absolutely regular processes themselves, but approximating functionals of such processes. In this situation, $L_2$ near epoch dependent processes are frequently considered. 
Since we also consider quantile-based estimators with the advantage of moment-freeness, we want to avoid moment assumptions implicitly in the short-range conditions. For this reason we employ the concept of near epoch dependence in probability, introduced by \citet{DVWW14}. For further information see \citet[][Appendix A]{DVWW14}.
\begin{definition} \label{def:pned}
We call the process $\left(X_n\right)_{n\in \N}$ near epoch dependent in probability (or short P-NED) on the process $\left(Z_n\right)_{n\in \mathbb{Z}}$ if there is a sequence of approximating constants $\left(a_l\right)_{l \in \N}$ with $a_l \rightarrow 0$ as $l\rightarrow \infty$, a sequence of functions $f_l:\mathbb{R}^{2l+1}\rightarrow \mathbb{R}$ and a non-increasing function $\Phi:\left(0,\infty\right)\rightarrow\left(0,\infty\right)$ such that
\begin{equation}
\p\left(\left|X_0-f_l\left(Z_{-l},\ldots,Z_l\right)\right|>\epsilon\right)\leq a_l\Phi\left(\epsilon\right)
\end{equation}
for all $l\in \N$ and $\epsilon>0$.
\end{definition} 
The absolute regularity coefficients $\beta_k$ and the approximating constants $a_l$ will have to fulfill certain rate conditions that are detailed in Assumption \ref{ass:gmd}.

%/////////////////////////////////////////////////////////////////////////////////////////////////////////////////////
%/////////////////////
%////////////
%/////////
%///////

\subsection{Change-point test statistics and long-run variance estimates}
\label{subsec:teststatistics}
 
We test the null hypothesis $H_0$ against the alternative $H_1$ be means of CUSUM-type test statistics of the form
\begin{equation}\label{general.teststatistic}
	\hat{T}_s = \max_{1\leq k \leq n}\frac{k}{\sqrt{n}}\left| s_{1:k} - s_{1:n} \right|,
\end{equation}
Throughout, we use $s_n$ as generic notation for a scale estimator (where we include, for completeness' sake, the variance), and $s_{1:n}$ denotes the estimator applied to $X_1,\ldots, X_k$. Considering, besides the variance, the four scale estimators introduced in the previous section, we have the test statistics $\hat{T}_{\sigma^2}$, $\hat{T}_d$, $\hat{T}_g$, $\hat{T}_m$, and $\hat{T}_Q(\alpha)$ based on the sample variance, the mean deviation, Gini's mean difference, the median absolute deviation and the $Q_n^{\alpha}$ scale estimator, respectively. Under the null hypothesis, the sequence $X_1,\ldots, X_n$ is stationary, and can be thought of as being part of a stationary process $(X_i)_{i\in\Z}$ with marginal distribution $F$
(i.e., $X_i = \lambda Y_i + \mu$).
Then, under suitable regularity conditions (that are specific to the choice of $s_n$), the test statistic $\hat{T}_s$ converges in distribution to $D_{s} \sup_{0\le t \le 1} |B(t)|$, where $B$ is a Brownian bridge.  The quantity $D^2_s$ is referred to as the long-run variance. It depends on the estimator $s_n$ and the data generating process. Expressions for the scale estimators considered here are given below. The distribution of $\sup_{0\le t \le 1}|B(t)|$ is well known and sometimes referred to as Kolmogorov distribution. However, $D^2_s$ is generally unknown, depends on the distribution of the whole process $(X_i)_{i\in\Z}$ and must be estimated when applying the test in practice.\footnote{Alternatively, bootstrapping can be employed, we take up this discussion in Section \ref{section_conclusion_and_outlook}.}

In the following definitions, let $X,Y \sim F$ be independent. The long-run variances corresponding to the scale estimators under consideration are
\[
	D^2_{\sigma^2} = \sum_{h=-\infty}^\infty \cov \left\{ (X_0 - E X_0)^2, (X_h-E X_h)^2 \right\}
\]
\[
	D^2_d = \sum_{h=-\infty}^\infty \cov \left( |X_0-\md(F) |, |X_h -\md(F)|\right),
\]
\be \label{eq:lrv.gmd}
	D^2_g = 4 \sum_{h=-\infty}^{\infty} \cov \left(\varphi(X_0),\varphi(X_h)\right),
\ee
where $\varphi(x) = E \left| x-Y\right| - g(F)$,
\[
	D^2_m = \frac{1}{f_Z^2(m(F))} \sum_{h =-\infty}^\infty \cov \left(\ind{|X_0-\md(F)| \le m(F)}, \ind{|X_h-\md(F)| \le m(F)}\right),
\]
where $f_Z$ is the density of $Z = |X-\md(F)|$, and 
\be \label{eq:lrv.qn}
	D^2_Q(\alpha) = \frac{4}{u^2(Q^\alpha(F))}\sum_{h=-\infty}^{\infty} 
	\cov \left( \psi(X_0),\, \psi(X_h)\right),
\ee
where $\psi(x)= P(|x-Y| \leq Q^\alpha) - \alpha$ and $u(t)$ is the density associated with the distribution function $U(t) = P(|X-Y|\leq t)$ of $|X-Y|$. 
An intuitive derivation of the expressions for $D^2_g$ and $D^2_Q(\alpha)$
%, which are not common knowledge and not easily accessible on first sight, 
are given in Appendix \ref{app_proofs}.

The following long-run variance estimators follow the construction principle of heterscedasticity and autocorrelation consistent (HAC) kernel estimators, for which we use results by \citet{dejong:2000}. The HAC kernel function (or weight function) $W$ can be quite general, but has to fulfill Assumption \ref{ass:W.b_n} (a) below, which is basically Assumption 1 of \citet{dejong:2000}. 
%The estimates generally differ very little with respect to the choice of the kernel. 
There is further a bandwidth to choose, which basically determines up to which lag the autocorrelations are included. For consistency of the long-run variance estimate, the bandwidth has to fulfill the rate condition of Assumption \ref{ass:W.b_n} (b).

\begin{assumption}\label{ass:W.b_n}
\mbox{ \\ } \begin{enumerate}[(a)]
\item
The function $W: \left[0,\infty \right)\rightarrow \left[-1,1\right]$ is continuous at $0$ and at all but a finite number of points and $W(0)=1$. Furthermore, $\left| W \right|$ is dominated by a non-increasing, integrable function and
\[
\int_{0}^{\infty}\left|\int_{0}^{\infty}W\left(t\right)\cos\left(xt\right)dt\right|dx < \infty.
\] 
\item
The bandwidth $b_n$ satisfies $b_n\rightarrow\infty$ as $n\rightarrow \infty$ and $b_n/\sqrt{n}\rightarrow 0$.
\een
\end{assumption}

We propose the following long-run variance estimators for the three moment-based scale measures: For the variance we take a weighted sum of empirical autocorrelations of the centered squares of the data, i.e.,  
\be \label{eq:lrv.cusum}
	\hat{D}^2_{\sigma^2,n} = 
	\sum_{k=-(n-1)}^{n-1}W \left(\frac{|k|}{b_n}\right)
	\frac{1}{n}\sum_{i=1}^{n-|k|} \left\{ (X_i - \bar{X}_n)^2-\sigma^2_n \right\} \left\{ (X_{i+|k|} -  \bar{X}_n)^2 - \sigma^2_n\right\}, 
\ee
where $\bar{X}_n$ and $\sigma^2_n$ denote the sample mean and the sample variance, respectively, computed from the whole sample, cf.\ Table~\ref{tab:scale_estimators}. Similar expressions have been considered, e.g., by \citet{Gombay_1996}, \citet{Lee_Park_2001} and \citet{Wied_Arnold_Bissantz_Ziggel_2012}. There are, of course, other possibilities regarding the exact definition the long-run variance estimate (e.g., use factor $1/n$ or $1/(n-k)$), the choice of which turns out to have an negligible effect. The strongest impact has the choice of the bandwidth $b_n$ (see Section \ref{subsec:bandwidth}). 

For the mean deviation, we propose
\[
	\hat{D}_{d}^2 = \sum_{k=-(n-1)}^{n-1} W \left(\frac{|k|}{b_n}\right)
	\frac{1}{n}\sum_{i=1}^{n-|k|}\left( |X_i-\md(F_n)|-d_n\right)\left(|X_{i+|k|}-\md(F_n)|-d_n\right),
\]
where $md(F_n)$ and $d_n$ denote the sample median and the sample mean deviation, respectively, of the whole sample (cf.\ Table~\ref{tab:scale_estimators}). For Gini's mean difference we consider
\[
	\hat{D}_{g}^2 = \, 4 \! \sum_{k=-(n-1)}^{n-1} W\left(\frac{|k|}{b_n}\right)
	\frac{1}{n}\sum_{i=1}^{n-|k|} \hat\varphi_n(X_i)\hat\varphi_n(X_{i+|k|}), 
\]
where 
\[
	\hat\varphi_n(x) =  \frac{1}{n}\sum_{i=1}^{n}\left|x-X_i\right| - g_n
\]
is an empirical version of $\varphi(x)$ in (\ref{eq:lrv.gmd}).
For the long-run variance estimates for the quantile-based scale measures $m_n$ and $Q_n^{\alpha}$, we need estimates for the densities $f_Z$ and $u$, respectively, for which we use kernel density estimates
\[
	\hat{f}(t) = \frac{1}{n h_n}\sum_{i=1}^n K\left(\frac{|X_i-\md(\hat{F})|-t}{h_n}\right),
\]
\[
	\hat{u}(t) = \frac{2}{n(n-1)h_n}\sum_{1\leq i<j\leq n}K\left(\frac{|X_i-X_j|-t}{h_n}\right).
\]
The density kernel $K$ and the bandwidth $h_n$ have to fulfill the following conditions.
\begin{assumption}\label{condition_K_d_n}
Let $K:\R \to \R$ be  symmetric (i.e.\ $K(x) = K(-x)$), Lipschitz-continuous function with bounded support which is of bounded variation and integrates to 1. Let the bandwidth $h_n$ satisfy $h_n\rightarrow 0$ and $n\, h_n^{8/3}\rightarrow \infty,$ as $n\rightarrow \infty$.
\end{assumption}
If we were interested in an accurate estimate of the whole density function, the estimate could be further improved by incorporating the fact that its support is the positive half-axis, but since we are only interested in density estimates at one specific point, this is not necessary. 
Define
\[
	D^2_{m} = \frac{1}{\hat{f}_Z(m_n)} \sum_{k=-(n-1)}^{n-1} W \left(\frac{|k|}{b_n}\right)
	\frac{1}{n}\sum_{i=1}^{n-|k|} \hat\xi_n(X_i)\hat\xi_n(X_{i+|k|})
\]
with
\[
	\hat\xi_n(x) = \Ind{|x-\md(F_n)|\le m_n} - 1/2, 
\]
and
\[
	D^2_{Q}(\alpha) = \frac{4}{\hat{u}(Q_n^\alpha)} \sum_{k=-(n-1)}^{n-1} W \left(\frac{|k|}{b_n}\right)
	\frac{1}{n}\sum_{i=1}^{n-|k|} \hat\psi_n(X_i)\hat\psi_n(X_{i+|k|})
\]
with
\[
	\hat\psi_n(x) = \frac{1}{n}\sum_{i=1}^{n} \Ind{|x-X_i|\le Q_n^{\alpha}} - \alpha. 
\]

Keep in mind that in the expressions above, the HAC bandwidth $b_n$ and the kernel density bandwidth $h_n$ play distinctively different roles: $b_n$ increases as $n$ increases, whereas $h_n$ decreases to zero. Also, the kernels $W$ and $K$ serve different purposes: $W$ is an HAC-kernel and is scaled such that $W(0)=1$, while $K$ is a density kernel and it is scaled such that it integrates to 1. 

Below, in Sections \ref{subsec:GMD} and \ref{subsec:qn}, we give sufficient conditions for the convergence of the studentized test statistics $\hat{D}^{-1}_g \hat{T}_g$ and $\hat{D}^{-1}_Q(\alpha) \hat{T}_Q(\alpha)$, respectively, since the corresponding estimators, as outlined in Section \ref{section_scale_estimators}, exhibit the best statistical properties, and these tests indeed show the best performance, as demonstrated in Section \ref{section_simulation_results}. 
The variance-based test statistic $\hat{D}^{-1}_{\sigma^2} \hat{T}_{\sigma^2}$, or versions of it, has been considered by several authors. It is treated for $L_2$ NED on $\alpha$-mixing processes by \citet{Wied_Arnold_Bissantz_Ziggel_2012}. As for the mean-deviation-based test statistic $\hat{D}^{-1}_d \hat{T}_d$, the convergence can be shown by similar techniques as for $\hat{D}^{-1}_{\sigma^2} \hat{T}_{\sigma^2}$: the same $(2+\delta)$ moment condition as for Gini's mean difference along with corresponding rate for the short-range dependence conditions (Assumption \ref{ass:gmd}) are required. Additionally, a smoothness condition around $\md(F)$ is necessary to account for the estimation of the central location. For the MAD-based test statistic $\hat{D}^{-1}_m \hat{T}_m$, no moment conditions are required, but smoothness conditions on $F$ at $\md(F)$ as well as $m(F) = |X-\md(F)|$, $X \sim F$, However, it turns out that the MAD does not provide a workable change-point test. Roughly speaking, the estimate is rather coarse, and the convergence to the limit distribution too slow to yield usable critical values. But even, for large $n$ or with the use of bootstrapping methods, the test is dominated in terms of power by the other tests considered.

%/////////////////////////////////////////////////////////////////////////////////////////////////////////////////////
%/////////////////////
%////////////
%/////////
%///////

\subsection{Gini's mean difference}
\label{subsec:GMD} 

We assume the following condition on the data generating process.

\begin{assumption}\label{ass:gmd}
Let $(X_i)_{i \in Z}$ be a stationary process that is $P$NED on an absolutely regular sequence $(Z_n)_{n\in\Z}$. There is 
a $\delta > 0$ such that
\begin{enumerate}[(a)]
\item
the $P$NED approximating constants $a_l$ and the absolute regularity coefficients $\beta_k$ satisfy
\be \label{eq:srd}
	a_l\Phi\left(l^{-6}\right) = O\left(l^{-6\frac{2+\delta}{\delta}}\right)
 \mbox{ as } l\rightarrow\infty \qquad \mbox{ and } \qquad 
 \sum_{k=1}^{\infty}k\beta_k^{\frac{\delta}{2+\delta}} < \infty,
\ee
where $\Phi$ is defined in Definition \ref{def:pned}, and 
\item
there is a positive constant $M$ such that $E| X_0|^{2+\delta} \leq M$ and $E \left| f_l\left(Z_{-l}, \ldots, Z_l\right)\right|^{2+\delta} \leq M$ for all $l \in \N$.
\end{enumerate}
\end{assumption}
Then we have the following result about the asymptotic distribution of the test statistic $\hat{D}^{-1}_g \hat{T}_g$ based on Gini's mean difference under the null hypothesis. The proof is given in Appendix \ref{app_proofs}.

\begin{theorem}\label{th:gmd}
Under Assumptions \ref{ass:W.b_n}, \ref{condition_K_d_n}, and \ref{ass:gmd}, we have 
$\hat{D}^{-1}_g \hat{T}_g \cid \sup_{0 \leq \lambda \leq 1}\left| B\left(\lambda\right)\right|$, where $\left(B\left(\lambda\right)\right)_{0 \leq \lambda \leq 1}$ is a standard Brownian bridge.
\end{theorem}

%/////////////////////////////////////////////////////////////////////////////////////////////////////////////////////
%/////////////////////
%////////////
%/////////
%///////

\subsection{$\boldsymbol{Q_n^{\alpha}}$}
\label{subsec:qn}

Since $Q_n^{\alpha}$ is a quantile-based estimator, we require no moment condition, and it suffices that the short-range dependence condition (\ref{eq:srd}) is satisfied for ``$\delta = \infty$''. 

\begin{assumption} \label{ass:qn.srd.rate}
Let $(X_i)_{i \in Z}$ be a stationary process that is $P$NED of an absolutely regular sequence $(Z_n)_{n\in\Z}$ such that 
the $P$NED approximating constants $a_l$ and the absolute regularity coefficients $\beta_k$ satisfy 
\[
	a_l\Phi(l^{-6}) = O(l^{-6}) 
	\ \mbox{ as  } l \rightarrow\infty \qquad \mbox{ and } \qquad 
	\sum_{k=1}^{\infty}k\beta_k < \infty,
\]
where $\Phi$ is defined in Definition \ref{def:pned}.
\end{assumption}

However, instead of the moment condition we require a smoothness condition on $F$.
\begin{assumption} \label{ass:smoothness.F}
The distribution $F$ has a Lebesgue density $f$ such that 
\begin{enumerate}[(a)]
\item
$f$ is bounded, 
\item
the support of $f$, i.e., $\overline{\{ x | f(x)  > 0 \}}$, is a connected set, and
\item
the real line can be decomposed in finitely many intervals such that $f$ is continuous and (non-strictly) monotonic on each of them. 
\end{enumerate}
\end{assumption}

We are now ready to state the following result concerning the asymptotic distribution of the $Q_n^{\alpha}$-based change-point test statistic. The proof is given in Appendix \ref{app_proofs}.

\begin{theorem}\label{th:qn}
Under Assumptions \ref{ass:W.b_n}, \ref{condition_K_d_n}, \ref{ass:qn.srd.rate}, and \ref{ass:smoothness.F}, we have, for any fixed $0 < \alpha < 1$ that
$\hat{D}^{-1}_{Q,n}(\alpha) \hat{T}_Q(\alpha) \cid \sup_{0 \leq \lambda \leq 1}\left| B\left(\lambda\right)\right|$,
where $\left(B\left(\lambda\right)\right)_{0 \leq \lambda \leq 1}$ is a standard Brownian bridge.
\end{theorem}

%/////////////////////////////////////////////////////////////////////////////////////////////////////////////////////
%/////////////////////
%////////////
%/////////
%///////
%\\\\\\\
%\\\\\\\\\
%\\\\\\\\\\\\
%\\\\\\\\\\\\\\\\\\\\\
%\\\\\\\\\\\\\\\\\\\\\\\\\\\\\\\\\\\\\\\\\\\\\\\\\\\\\\\\\\\\\\\\\\\\\\\\\\\\\\\\\\\\\\\\\\\\\\\\\\\\\\\\\\\\\\\\\\\\

\subsection{Data-adaptive bandwidth selection}
\label{subsec:bandwidth}

The rate conditions of Assumptions \ref{ass:W.b_n} on the HAC bandwidth $b_n$ to achieve consistency of the long-run variance estimate are rather mild. However, the questions remains  how to choose $b_n$ optimally for a given sequence of observations of length $n$. The answer depends on the degree of serial dependence present in the sequence. Loosely speaking, choosing it too small results in a size distortion, choosing it too large will render the tests conservative and less powerful. %The choice of $b_n$ tends to have the largest influence on the long-run variance estimates as compared to the kernel density bandwidth $h_n$ or kernels $W$ and $K$. 

A general approach is to assume a parametric time-series model, minimize the mean squared error in terms of the parameters of the model and then plug-in estimates for the parameters. For instance, to estimate the long-run variance of an AR(1) process with auto-correlation parameter $\rho$
\citet[][Section 6]{andrews:1991} gives an optimal bandwidth of
\be \label{eq:optimal.bandwidth}
	b_n = 1.447 n^{1/3} \left(\frac{4 \rho^2}{(1-\rho^2)^2 }\right)^{1/3}
\ee
when the Bartlett kernel is used. 

When employing alternative methods to estimate the long-run variance, i.e., block sub-sampling or block bootstrap, one faces the similar challenge of selecting an appropriate blocklength, and techniques that have been derived for data-adaptive blocklength selection may be ``borrowed'' for bandwidth selection. For instance, \citet{Carlstein.1986} obtains a very similar expression as (\ref{eq:optimal.bandwidth}) for an optimal blocklength in the AR(1) setting for a block sub-sampling scheme. 

\citet{Politis.2004} (see also \citet{Patton.2009}) derive optimal blocklengths for block bootstrap methods and obtain also an expression similar to (\ref{eq:optimal.bandwidth}). This technique is adapted, e.g., by \citet{Kojadinovic.2015} for a change-point test based on Spearman's rho. Furthermore, \citet{Bucher.2016} identified the asymptotically optimal blocklength for the multiplier bootstrap of U-statistics and developed a method to estimate it.

All such methods require the estimation of unknown parameters as, e.g., in (\ref{eq:optimal.bandwidth}) the autocorrelation. Alongside exploring alternatives to the sample variance for robust scale estimation, it is advisable to consider alternatives to the sample autocorrelation. A recent comparison of robust autocorrelation estimators is given by \citet{Duerre.2015}, who recommend to use Gnanadesikan-Kettenring-type estimators. Such estimates require robust scale estimates, for which also the $Q_n$ has been proposed \citep{ma:genton:2000}.

\citet{hall1995blocking} have suggested an alternative approach without estimating autocorrelations: If the rate, but not the constant, in an expression as (\ref{eq:optimal.bandwidth}) is known, the block length can be chosen by sub-sampling. For $U$-statistics, the rate is $n^{1/3}$ in the case of the Bartlett kernel \citep[][]{dehling2010central} and $n^{1/5}$ for flat-top kernels \citep[][]{Bucher.2016}. Results on the optimal block length for quantiles and U-quantiles are not known to us, but we suspect that the rates for the optimal block length are the same as for partial sums, because our variance estimator uses an linearization.

%/////////////////////////////////////////////////////////////////////////////////////////////////////////////////////
%/////////////////////
%////////////
%/////////
%///////
%\\\\\\\
%\\\\\\\\\
%\\\\\\\\\\\\
%\\\\\\\\\\\\\\\\\\\\\
%\\\\\\\\\\\\\\\\\\\\\\\\\\\\\\\\\\\\\\\\\\\\\\\\\\\\\\\\\\\\\\\\\\\\\\\\\\\\\\\\\\\\\\\\\\\\\\\\\\\\\\\\\\\\\\\\\\\\

\section{The choice of $\alpha$ for $Q_n^\alpha$}
\label{section_qn}

%\color{blue}

In this section, we address the question which $\alpha$ to choose best when employing the $Q_n^\alpha$ scale estimator.
% in the current change-point setting. 
The theoretical U-quantile results of the previous setting apply to $Q_n^\alpha$ for any $0 < \alpha < 1$. 
The original $Q_n$ proposed by \citet{Rousseeuw_Croux_1993} corresponds roughly to $\alpha = 1/4$, which is mainly motivated by breakdown considerations. 
A high-breakdown-estimator is in fact counterproductive for change-point detection purposes, and higher values of $\alpha$ are more appropriate for the problem at hand.
A high-breakdown-point estimator is designed, in some sense, to ignore a large portion of outliers, no matter how they are distributed, spatially or temporarily. 
The perception of robustness in the change-point setting is conceptually different: we want to safeguard against a few outliers or several but evenly distributed over the observed sequence, as they may be generated by a heavy-tailed distribution. A subsequence of outliers on the other hand, which even exhibits some common characteristics differing from the rest of the data, constitutes a structural change, which we want to detect rather than ignore. 
To illustrate this point consider the following example: we have a sample where the first 60\% of the observations alternate between $1$ and $-1$ and  the second 40\% alternate between $100$ and $-100$. This constitutes a very noticeable change in the scale, but the $Q_n^{0.25}$-based change-point test of the CUSUM type is not able to detect this change. 
Similar counter-examples can be constructed for analogous change-point tests which compare the first part to the second part of the sample. 
However, this very artificial example should not be given to much attention. Any quantile-based method is generally suitable for continuous distributions only. 

The point may also be illustrated by Monte Carlo simulation results for continuous distributions. Consider the following data situation: starting from an i.i.d.\ sequence of standard normal observations, we multiply the second half by some value $\lambda$. Table \ref{tab:qn.schlecht} shows rejection frequencies at the asymptotic $5\%$ level based on the original $Q_n$ for sample sizes $n=60$ and $n=500$ and several values of $\lambda$ (bandwidth as in the simulations in Section \ref{section_simulation_results}, i.e., $b_n = 2n^{1/3}$). 
\begin{table}
\small
\caption{Power of change-point tests at asymptotic 5\% level, based on sample variance/Gini's mean difference/$Q_n$ for independent, centered normal observations. Standard deviation changes from 1 in first half to $\lambda$ in second half.  \label{tab:qn.schlecht}}
%\hspace{-5.0ex}
\begin{center}
\begin{tabular}{c|ccccccc}
 & \multicolumn{5}{c}{standard deviation $\lambda$ in second half }\\
sample size & $1.0$ &  $1.2$ & $1.5$ & $2.0$ & $3.0$ \\
\hline
$n=60$	  &   .04/.04/.44  & .04/.07/.32 & .12/.24/.18 & .23/.51/.06 & .36/.82/.01\\
$n=500$   &   .04/.04/.08 &  .63/.64/.52 & 1/1/1 & 1/1/1 & 1/1/1 \\
\end{tabular}
\end{center}
\end{table}
The numbers of Table \ref{tab:qn.schlecht} suggest that the original $Q_n$ is very unsuitable for change-point detection in small samples. For $n=60$, it strongly exceeds the size, and the power is actually lower under alternatives than under the null. 

In principle, when the data is large enough, the $Q_n^{0.25}$ is able to pick up scale changes in the continuous setting, i.e., the $1/4$th quantile of all pairwise differences is larger than the $1/4$th quantile of all pairs of the first half of the observations. For $n=500$, this difference is sufficiently pronounced so that the test more or less works satisfactorily. However, for $n=60$, this difference is relatively small compared to the increased long-run variance (by which the test statistic is divided) when $\lambda$ is large. This largely accounts for the decrease of power as $\lambda$ increases. Additionally, the $Q_n$-based test grossly exceeds the size under the null. This can be largely attributed to a general bad ``small sample behavior'' of the test statistic. We observe a similar effect for the MAD-based test, cf.\ Section \ref{section_simulation_results} and the median-based change-point test for location considered by \citet{VW14}. Thus the $Q_n^{0.25}$ is particularly unsuitable for a quick detection of strong changes.

The problems can be overcome by considering $Q_n^\alpha$ for larger values of $\alpha$ than $1/4$, and using such a scale estimator indeed leads to a workable change-point test also for $n = 60$.
%
%
%
%\footnote{They specifically define the finite-sample version as the $\binom{\lfloor n/2 \rfloor + 1}{2}$th order statistic of the $\binom{n}{2}$ values $|X_i - X_j|$, $1 \le i < j \le n$.} and report an asymptotic relative efficiency with respect to the standard deviation at normality of $82\%$. The main motivation for their choice is to achieve the maximal breakdown point of about 50\%. This aspect is of much lesser relevance in the change-point context, quite the contrary, the very property of high-breakdown-point robustness is counterproductive for detecting change points. A larger value of $\alpha$ is much more appropriate for the problem at hand, as we will detail below. 
%
\begin{figure}
\centering
\includegraphics[width=\textwidth]{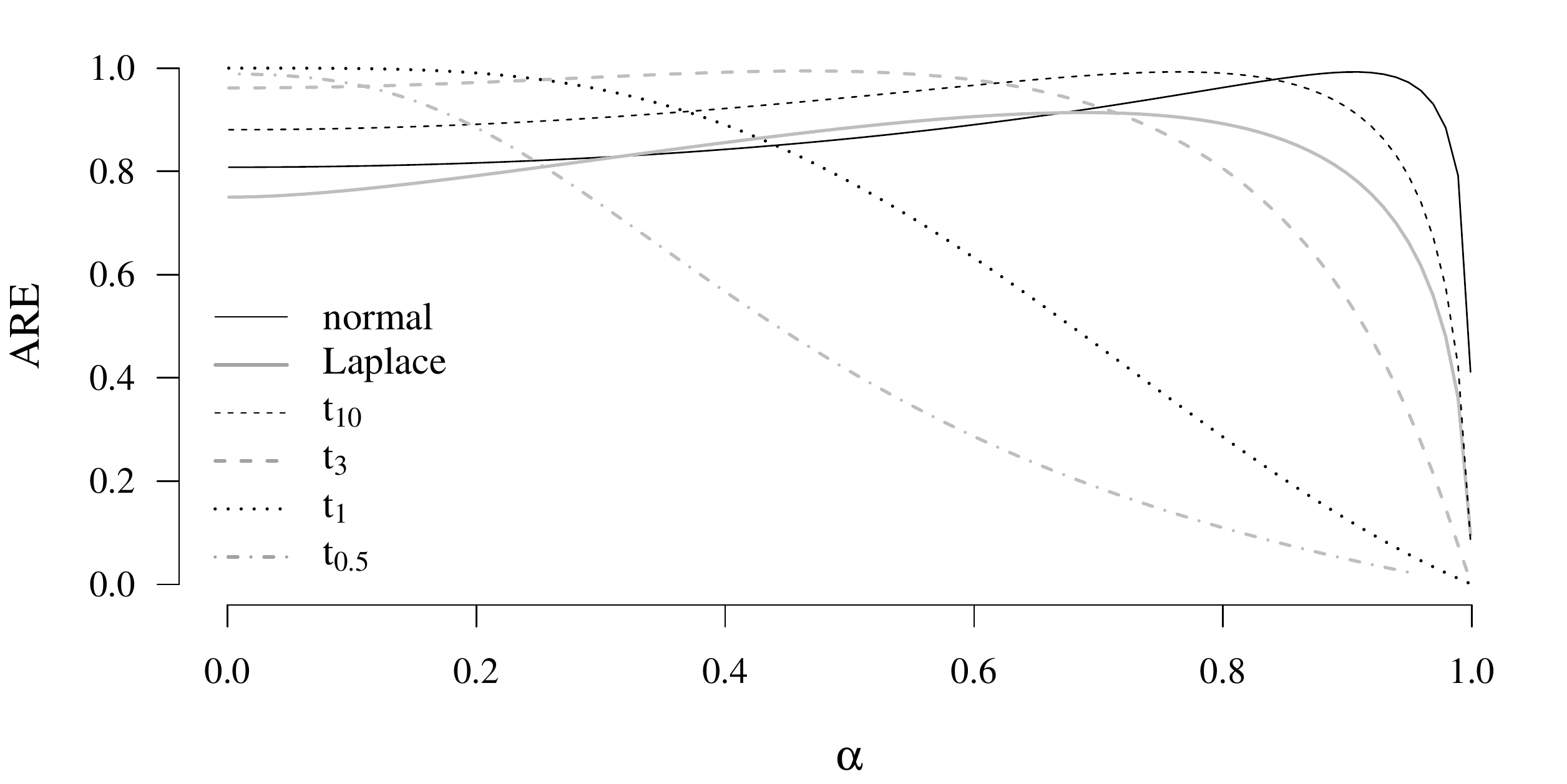}
\caption{Asymptotic relative efficiencies (AREs) of $Q_n^{\alpha} $ at normal, Laplace, and several $t$ distributions wrt respective maximum-likelihood estimators of scale.\label{plot_qn_are_alpha}}
\end{figure}
As a guideline for a suitable choice of $\alpha$, we may look at the asymptotic efficiencies. Figure \ref{plot_qn_are_alpha} plots the asymptotic relative efficiency of $Q_n^\alpha$ at several scale families with respect to the respective maximum-likelihood estimator for scale. The solid line, showing the asymptotic relative efficiency of $Q_n^\alpha$ with respect to the standard deviation at normality, is also depicted in \citet{Rousseeuw_Croux_1992}. Since we are also interested in efficiency at heavy-tailed distributions, we further include several members of the $t_\nu$-family (for $\nu = 1/2, 1, 3, 10$) and the Laplace distribution, cf.~Table~\ref{tab:distributions}. The mathematical derivations for this plot, i.e., the asymptotic variances of the $Q_n^\alpha$ and the MLE of the scale parameter of the $t$-distribution, are given in Appendix \ref{app_supp}.

At normality, the $Q_n^{\alpha}$ is asymptotically most efficient for $\alpha = 0.9056$ with an asymptotic relative efficiency of 99\% with respect to the standard deviation, but it shows a very slow decay as $\alpha$ decreases to zero. For the $t$-distributions considered, the maximal asymptotic efficiency of the $Q_n^\alpha$ is achieved for smaller values of $\alpha$, e.g., for the $t_3$ distribution, the optimum is attained at $\alpha = 0.4661$. 
However, within the range $0 < \alpha < 0.73$, the asymptotic relative efficiency with respect to the maximum-likelihood estimator is above 90\%. The $t_\nu$-distributions with $\nu = 1/2$ and $\nu = 1$, also depicted in Figure~\ref{plot_qn_are_alpha}, are extremely heavy-tailed.  We consider $t_\nu$ distributions with $\nu= 3$ and $\nu = 5$ as more realistic data models, and these are also included in the simulation results presented in Section \ref{section_simulation_results}.

Altogether, Figure \ref{plot_qn_are_alpha} suggests that $\alpha \approx 3/4$ may be a suitable choice as far as asymptotic efficiency is concerned.
We ran simulations with many different values of $\alpha$ and found that $Q_n^\alpha$ generally
performed best within that the range $0.7 < \alpha < 0.9$. In the tables in Section \ref{section_simulation_results}, we report results for $\alpha = 0.8$.

%/////////////////////////////////////////////////////////////////////////////////////////////////////////////////////
%/////////////////////
%////////////
%/////////
%///////
%\\\\\\\
%\\\\\\\\\
%\\\\\\\\\\\\
%\\\\\\\\\\\\\\\\\\\\\
%\\\\\\\\\\\\\\\\\\\\\\\\\\\\\\\\\\\\\\\\\\\\\\\\\\\\\\\\\\\\\\\\\\\\\\\\\\\\\\\\\\\\\\\\\\\\\\\\\\\\\\\\\\\\\\\\\\\\

\section{Simulation Results}
\label{section_simulation_results}

In a simulation study we want to investigate the empirical size and power of the change-point tests based on the test statistics
introduced in Section \ref{subsec:teststatistics}. We consider the following general model: We assume the process $(Y_i)_{i \in\Z}$ to follow an $AR(1)$ process
\begin{equation*}
Y_i = \rho Y_{i-1}+\epsilon_i, \qquad i \in \Z,
\end{equation*}
for some $-1  < \rho < 1$
where the $\epsilon_i$ are an i.i.d.\ sequence with a mean-centered distribution. Then the data $X_1, \ldots, X_n$ are generated as
\begin{equation}
X_i = \begin{cases}
				Y_i,	& \quad 1 \leq i \leq \lfloor \theta n \rfloor, \\
\lambda Y_i,	& \quad \lfloor \theta n \rfloor < i \leq n,
\end{cases}
\end{equation}
for some $0 < \theta < 1$. Thus we have four parameters, $\lambda$, $\theta$, $\rho$, and $n$, which regulate the size of the change, the location of the change, the degree of the serial dependence of the sequence and the sample size, respectively. The tables below report results for $\lambda = 1, 1.5, 2$, $\theta = 1/4, 1/2, 3/4$, $\rho = 0, 0.8$ and $n = 60, 120, 240, 500$. A fifth ``parameter'' is the distribution of the $\epsilon_i$, for which we consider the standard normal $N(0,1)$, the standard Laplace $L(0,1)$, the normal scale mixture $NM(\gamma, \varepsilon)$ with $\gamma = 3$, $\varepsilon = 0.01$ and $t_\nu$-distributions with $\nu = 3$ and $\nu = 5$. The densities of these distributions are summarized in Table~\ref{tab:distributions}. The latter four of them are, to varying degrees, heavier-tailed than the normal. 
\begin{table}
\caption{
		Densities of innovation distributions considered in the simulation set-up. $B(\cdot, \cdot)$ denotes the beta function.
		\label{tab:distributions}
}
%\newcolumntype{C}[1]{>{\centering}m{#1}}
\small
%\newdimen\mylength
%\mylength=0.33\textwidth
\renewcommand{\arraystretch}{2.5}
\centering
\begin{tabular}{p{0.17\textwidth}|c|p{0.11\textwidth}|c}
\hline
distribution & density $f(x)$ & parameters & kurtosis\\
\hline
		normal $N(\mu,\sigma^2)$ & 
		$\ds \frac{1}{\sqrt{2\pi \sigma^2} } \exp\left\{ - \frac{(x-\mu)^2}{2 \sigma^2} \right\}$ &
		 $\mu \in \R$, \newline \centering $\sigma^2 > 0$  &
		0 \\
\hline	 
		normal mixture \newline $NM(\gamma,\epsilon)$ & 
		$\ds \frac{\epsilon}{\sqrt{2\pi}\gamma}\exp\left\{\frac{-x^2}{2\gamma^2}\right\}+\frac{1-\epsilon}{\sqrt{2\pi}}\exp\left\{\frac{-x^2}{2}\right\}$ &
		$0 \le \epsilon \le 1$, \newline \centering $\gamma \ge 1$ &
		$\ds \frac{3 \epsilon(1-\epsilon) (\lambda^2-1)^2}{(\epsilon \lambda^2 +1 - \epsilon)^2}$ \\[0.5ex]		
\hline		
		Laplace $L(\mu,\alpha)$ & 	
		$\ds \frac{1}{2\alpha}\exp\left\{\frac{-|x-\mu|}{\alpha}\right\}$ &
		\centering  $\mu \in \R$, \newline $\alpha >0$ &
		3 \\
\hline		
		$t_\nu$ & 	
		$\ds \sqrt{\nu} B(\nu/2,1/2) \left( 1 + x^2/\nu\right)^{-\frac{\nu+1}{2}}$  &
		\centering $\nu > 0 $ & 
		$\ds \frac{6}{\nu-4}$ ($\nu > 4$) \\[0.5ex]
\hline	 
\end{tabular}
\end{table}
The normal mixture distribution captures the notion that the majority of the data stems from the normal distribution, except for some small fraction, which stems from a normal distribution with a $\gamma$ times larger standard deviation. This type of contamination model has been popularized by \citet{Tukey_1960}, who argues that $\gamma = 3$ is a realistic choice in practice, and furthers pointed out that in this case the mean deviation $d_n$ is more efficient scale estimator than the standard deviation for values of $\epsilon$ as low as 1\%.
Concerning the long-run variance estimation, we take the following choices for bandwidths and kernels, 
\be \label{eq:choices}
	K(t) = \frac{3}{4} (1- t^2) \varind{[-1,1]}(t), 
	\ \ 
	W(t) = (1- t^2)^2 \varind{[-1,1]}(t),
%\ee\[
  \  \
	h_n = I_n\, n^{-1/3}, 
	\ \ 
	b_n = 2 n^{1/3},
\ee
where $I_n$ denotes the sample interquartile range of the data points the kernel density estimator is applied to. The kernel $K$ above is known as Epanechnikov kernel, and $W$ as quartic kernel. The results of the long-run variance estimation generally differ very little with respect to the choice of kernels. 
The the HAC-bandwidth $b_n$ tends to have the largest influence. We fix it in our simulations to $b_n = 2 n^{1/3}$, which is large enough to account for the strong serial dependence setting of $\rho = 0.8$, but is certainly far from optimal in the independence setting, and the results can be improved upon by using one of the data-adaptive bandwidth selection methods outlined in Section \ref{subsec:bandwidth}. However, choosing the same, albeit fixed, bandwidth for all change-point tests allows a fair comparison and puts the focus on the impact of the different estimators. 

%=================================================================================================================
% table size
%=================================================================================================================

\begin{table}[t]
\small
\caption{\label{table.size} 
   {\it Test size.} Rejection frequencies (\%)  of change-point tests based on six different scale estimators at stationary sequences. Asymptotic 5\% significance level; sample sizes $n = 60, 120, 240, 500$; five different innovation distributions. HAC bandwidth for long-run variance estimation: $b_n = 2n^{1/3}$.
}
\smallskip
\centering
\begin{tabular}{c|@{\,\,\,}r@{\,\,\,}r@{\,\,\,}r@{\,\,\,}r@{\,\,\,}r@{\,\,\,}r|r@{\,\,\,}r@{\,\,\,}r@{\,\,\,}r@{\,\,\,}r@{\,\,\,}r}
\hline
					  	& \multicolumn{6}{c}{independence} & \multicolumn{6}{|c}{$AR(1)$ with $\rho=0.8$}\\\hline
 Estimator	&  Var & MD &   GMD  & MAD& $Q_n$ & $Q_n^{0.8}$ &  Var & MD & GMD & MAD& $Q_n$ & $Q_n^{0.8}$\\\hline\hline
 \multicolumn{1}{c}{ $n=60$	}	&	\multicolumn{6}{c}{}	&	\multicolumn{6}{c}{} \\\hline
$N(0,1)$&		3&5&3&27&43&6 	&8&7&13&24&15&14  \\
$L(0,1)$&		3&4&3&23&38&8	&7&6&11&23&12&12\\
$NM(3,0.01)$&	4&6&4&26&46&7	&6&6&11&22&14&10\\
$t_5$&			3&4&1&24&40&8  	&7&6&11&23&13&13\\ 
$t_3$&			2&5&2&27&43&11	&4&4&8&24&12&11\\\hline\hline
 \multicolumn{1}{c}{  $n=120$	}	&	\multicolumn{6}{c}{}	&	\multicolumn{6}{c}{} \\\hline
$N(0,1)$&		3&3&3&21&29&4	&6&6&8&19&8&7\\   
$L(0,1)$&		2&2&2&15&23&5	&3&4&7&17&7&7\\
$NM(3,0.01)$&	2&2&2&19&27&4	&4&4&6&18&8&7\\
$t_5$&			1&2&2&19&23&4	&4&4&7&17&8&7\\     
$t_3$&			2&4&2&16&25&8	&3&3&5&19&6&6\\\hline\hline
 \multicolumn{1}{c}{  $n=240$	}	&	\multicolumn{6}{c}{}	&	\multicolumn{6}{c}{} \\\hline
$N(0,1)$&		2&3&2&15&11&2	&3&3&5&15&5&4\\   
$L(0,1)$&		2&3&2&12&10&3	&3&3&5&13&5&5\\
$NM(3,0.01)$&	2&3&3&14&11&4	&4&4&6&14&5&6\\
$t_5$&			3&4&3&13&12&3	&2&4&5&14&5&4\\    
$t_3$&			2&3&3&13&13&6	&2&3&4&15&4&5\\\hline\hline
  \multicolumn{1}{c}{ $n=500$	}	&	\multicolumn{6}{c}{}	&	\multicolumn{6}{c}{} \\\hline
$N(0,1)$&		3&3&3&10&5&4	&4&4&4&11&5&5\\   
$L(0,1)$&		3&5&4&9&6&5		&5&6&5&10&4&4\\ 
$NM(3,0.01)$&	2&3&3&12&6&3	&4&4&5&12&3&3\\
$t_5$&			3&4&3&11&7&4	&4&4&4&10&5&4\\     
$t_3$&			1&3&2&11&7&5	&3&4&3&12&4&4 \\
\hline
\end{tabular}
\end{table}

For each setting we generate 1000 repetitions. Tables~\ref{table.size}--\ref{table.power.dep.2} report empirical rejection frequencies (in \%) at the asymptotic 5\% significance level, i.e., we count how often the test statistics exceed 1.358, i.e., the 95\%-quantile of the limiting distribution of the studentized test statistics under the null.

%=================================================================================================================
% tables power, independent, lambda = 1.5
%=================================================================================================================    
\begin{table}[t]
\small
\caption{ \label{table.power.ind.1_5}
	{\it Test power.} Rejection frequencies (\%) at asymptotic 5\% level.  
	Change-point tests based on variance (Var), mean deviation (MD), Gini's mean difference (GMD), and $Q_n^{0.8}$.
	{\bf Independent observations};  scale changes by {\bf factor $\lambda = 1.5$}.
	Several sample sizes, marginal distributions, and change locations. HAC bandwidth: $b_n = 2n^{1/3}$.
}

\smallskip
\centering
\begin{tabular}{c|r@{\,\,\,}r@{\,\,\,}r@{\,\,\,}r|r@{\,\,\,}r@{\,\,\,}r@{\,\,\,}r|r@{\,\,\,}r@{\,\,\,}r@{\,\,\,}r} 
\hline
 Change location: & \multicolumn{4}{|c}{$[n/4]$}& \multicolumn{4}{|c}{$[n/2]$} & \multicolumn{4}{|c}{$[n3/4]$}\\\hline
 Estimator:	&  Var & MD & GMD & $Q_n^{0.8}$ &  Var & MD & GMD & $Q_n^{0.8}$ &  Var & MD & GMD & $Q_n^{0.8}$ \\\hline\hline
\multicolumn{1}{c}{$n=60$}		&	\multicolumn{4}{c}{}	&	\multicolumn{4}{c}{} &	\multicolumn{4}{c}{}	\\\hline
	$N(0,1)$		&4&3&10&11				&9&8&21&22			&5&4&8&9\\
	$L(0,1)$		&2&2&5&8				&4&4&10&14			&2&3&5&8\\ 
	$NM(3,0.01)$	&4&4&8&12				&4&4&10&20			&5&4&10&10\\
	$t_5$			&2&4&6&9				&4&6&11&12			&3&4&6&8\\     
	$t_3$ 			&1&2&3&8				&3&5&8&12			&3&4&4&10\\\hline\hline
 \multicolumn{1}{c}{$n=120$	}	&	\multicolumn{4}{c}{}	&	\multicolumn{4}{c}{} &	\multicolumn{4}{c}{}	\\\hline
$N(0,1)$		&10&12&22&22			&43&43&57&49		&26&21&34&24\\
    $L(0,1)$		&4&6&11&8				&13&21&27&21		&11&11&15&11\\ 
    $NM(3,0.01)$	&6&11&17&18				&30&37&47&46		&22&19&30&20\\
    $t_5$			&4&8&11&10				&18&28&34&28		&11&12&19&13\\     
    $t_3$ 			&2&5&7&7				&8&19&22&20			&7&11&13&10\\\hline\hline
 \multicolumn{1}{c}{$n=240$	}	&	\multicolumn{4}{c}{}	&	\multicolumn{4}{c}{} &	\multicolumn{4}{c}{}	\\\hline
$N(0,1)$		&36&46&60&57			&90&88&93&92		&72&64&77&69\\
    $L(0,1)$		&11&21&26&22			&44&62&65&56		&34&39&44&32\\ 
    $NM(3,0.01)$	&25&37&49&49			&72&86&89&89		&57&59&69&65\\
    $t_5$			&12&26&31&30			&48&73&74&70		&34&48&50&44\\     
    $t_3$ 			&5&18&18&20				&19&52&49&51		&15&32&32&32\\\hline\hline
 \multicolumn{1}{c}{$n=500$	}	&	\multicolumn{4}{c}{}	&	\multicolumn{4}{c}{} &	\multicolumn{4}{c}{}	\\\hline
$N(0,1)$		&92&94&97&96			&100&100&100&100	&99&97&98&99\\
    $L(0,1)$		&40&70&72&64			&88&96&95&94		&75&83&84&79\\ 
    $NM(3,0.01)$	&66&90&93&94			&95&100&100&100		&89&96&98&99\\
    $t_5$			&39&78&77&78			&85&99&98&99		&73&90&90&91\\     
    $t_3$ 			&10&53&46&57			&40&89&84&93		&34&70&64&74 \\
\hline
\end{tabular}
\end{table}

%=================================================================================================================
% tables power, independent, lambda = 2
%================================================================================================================= 
\begin{table}[t]
\small
\caption{ \label{table.power.ind.2}
  {\it Test power.} Rejection frequencies (\%) at asymptotic 5\% level.  
	Same set-up as Table~\ref{table.power.ind.1_5} except scale changes by {\bf factor $\lambda = 2.0$}. 
}
\smallskip
\centering
\begin{tabular}{c|r@{\,\,\,}r@{\,\,\,}r@{\,\,\,}r|r@{\,\,\,}r@{\,\,\,}r@{\,\,\,}r|r@{\,\,\,}r@{\,\,\,}r@{\,\,\,}r} 
\hline
Change location:  & \multicolumn{4}{|c}{$[n/4]$}& \multicolumn{4}{|c}{$[n/2]$} & \multicolumn{4}{|c}{$[n3/4]$}\\\hline
 Estimator:	&  Var & MD & GMD & $Q_n^{0.8}$ &  Var & MD & GMD & $Q_n^{0.8}$ &  Var & MD & GMD & $Q_n^{0.8}$ \\
		\hline\hline
		 \multicolumn{1}{c}{$n=60$	}	&	\multicolumn{4}{c}{}	&	\multicolumn{4}{c}{} &	\multicolumn{4}{c}{}	\\\hline
		$N(0,1)$		&2&4&18&21				&18&18&46&39		&7&5&14&9\\
		$L(0,1)$		&1&1&7&12				&7&8&21&22			&5&3&8&10\\ 
		$NM(3,0.01)$	&3&4&16&21				&15&16&40&32		&6&4&14&9\\
		$t_5$			&2&3&9&12				&8&12&29&24			&4&3&8&11\\     
		$t_3$ 			&1&2&5&10				&5&7&18&19			&4&3&8&12\\\hline\hline
		 \multicolumn{1}{c}{$n=120$	}	&	\multicolumn{4}{c}{}	&	\multicolumn{4}{c}{} &	\multicolumn{4}{c}{}	\\\hline
		$N(0,1)$		&17&27&53&52			&79&86&94&84		&58&57&75&40\\
		$L(0,1)$		&6&11&22&22				&34&58&68&47		&25&30&42&20\\ 
		$NM(3,0.01)$	&11&20&43&45			&65&80&88&81		&47&49&66&37\\
		$t_5$			&6&14&26&26				&43&68&75&63		&31&40&49&28\\     
		$t_3$ 			&3&8&14&16				&18&46&51&45		&15&29&35&22\\\hline\hline
		 \multicolumn{1}{c}{$n=240$	}	&	\multicolumn{4}{c}{}	&	\multicolumn{4}{c}{} &	\multicolumn{4}{c}{}	\\\hline
		$N(0,1)$		&72&89&97&94			&100&100&100&100	&99&99&100&98\\
		$L(0,1)$		&24&57&66&54			&86&98&98&94		&76&86&89&77\\ 
		$NM(3,0.01)$	&52&83&92&92			&94&100&100&100		&91&98&99&97\\
		$t_5$			&26&70&76&77			&85&99&99&98		&77&92&93&90\\     
		$t_3$ 			&9&45&45&49				&46&92&89&94		&43&76&75&78\\\hline\hline
		 \multicolumn{1}{c}{$n=500$	}	&	\multicolumn{4}{c}{}	&	\multicolumn{4}{c}{} &	\multicolumn{4}{c}{}	\\\hline
		$N(0,1)$		&100&100&100&100		&100&100&100&100	&100&100&100&100\\
		$L(0,1)$		&81&100&100&98			&100&100&100&100	&99&100&100&100\\ 
		$NM(3,0.01)$	&91&100&100&100			&100&100&100&100	&100&100&100&100\\
		$t_5$			&76&100&100&100			&98&100&100&100		&97&100&100&100\\     
		$t_3$ 			&30&94&89&98			&75&100&99&100		&72&98&98&100 \\
	  \hline
\end{tabular}
\end{table}                 

\paragraph{Analysis of size.} Table~\ref{table.size} reports rejection frequencies of the change-point tests based on the variance $\sigma_n^2$ (Var), the mean deviation $d_n$ (MD), Gini's mean difference $g_n$ (GMD), the median absolute deviation $m_n$ (MAD), the original $Q_n$ as considered by \citet{Rousseeuw_Croux_1993}, i.e., the $\binom{\lfloor n/2 \rfloor + 1}{2}$-th order statistic of
% $|X_i - X_j|$, $1 \le i < j \le n$, 
all pairwise differences, and the $Q_n^{0.8}$, i.e., the $0.8 \binom{n}{2}$-th order statistic of 
%$|X_i - X_j|$, $1 \le i < j \le n$.
all pairwise differences.
We notice that the $Q_n$ and the MAD heavily exceed the size. This effect wears off as $n$ increases, but rather slowly. The $Q_n$ shows an acceptable size behavior for $n=500$, the MAD not even for this sample size. 
This behavior can be attributed to a discretization effect. 
A similar observations is made at the median-based change-point test for location, which is discussed in detail in \citet[][Section 4]{VW14}. Due to the size distortion, the MAD and the $Q_n$ are excluded from any further power analysis. The size exceedance also takes places to a lesser degree for the $Q_n^{0.8}$ at $n = 60$ and for GMD under dependence at $n = 60$. This must be taken into account when comparing the power under alternatives. Based on the simulation results for size and power, the $Q_n^{0.8}$ can be seen to provide a sensible change-point test for scale, but some caution should be taken for sample sizes below $n=100$.  All other tests keep the size for the situations considered.

%=================================================================================================================
% table power, dependent, lambda = 1.5
%=================================================================================================================     
\begin{table}[t]
\small
\caption{ \label{table.power.dep.1_5}
	{\it Test power.} Rejection frequencies (\%) at asymptotic 5\% level.  
	Change-point tests based on variance (Var), mean deviation (MD), Gini's mean difference (GMD), and $Q_n^{0.8}$.
	{\bf AR(1) process with $\rho=0.8$};  scale changes by {\bf factor $\lambda = 1.5$}.
	Several sample sizes, marginal distributions, and change locations. HAC bandwidth: $b_n = 2n^{1/3}$.
}

\smallskip
\centering
\begin{tabular}{c|r@{\,\,\,}r@{\,\,\,}r@{\,\,\,}r|r@{\,\,\,}r@{\,\,\,}r@{\,\,\,}r|r@{\,\,\,}r@{\,\,\,}r@{\,\,\,}r} 
\hline
Change location: & \multicolumn{4}{|c}{$[n/4]$}& \multicolumn{4}{|c}{$[n/2]$} & \multicolumn{4}{|c}{$[n3/4]$}\\\hline
 Estimator:	&  Var & MD & GMD & $Q_n^{0.8}$ &  Var & MD & GMD & $Q_n^{0.8}$ &  Var & MD & GMD & $Q_n^{0.8}$ \\\hline\hline
\multicolumn{1}{c}{$n=60$}		&	\multicolumn{4}{c}{}	&	\multicolumn{4}{c}{} &	\multicolumn{4}{c}{}	\\\hline
$N(0,1)$		&11&10&20&20			&17&13&30&25		&16&11&24&19\\
$L(0,1)$		&10&8&19&18				&11&8&21&20			&12&7&18&18\\
$NM(3,0.01)$	&9&9&19&20				&15&12&25&25		&15&8&23&20\\
$t_5$			&9&9&18&20				&14&12&23&24		&12&9&20&19\\    
$t_3$ 			&8&6&15&17				&10&10&21&23		&11&7&18&17\\\hline\hline
\multicolumn{1}{c}{$n=120$}		&	\multicolumn{4}{c}{}	&	\multicolumn{4}{c}{} &	\multicolumn{4}{c}{}	\\\hline
$N(0,1)$		&8&10&18&19				&20&22&36&30		&16&14&25&19\\
$L(0,1)$		&7&8&15&15				&18&17&30&24		&12&11&21&13\\
$NM(3,0.01)$	&10&11&18&18			&18&20&32&24		&15&14&24&16\\
$t_5$			&8&8&16&15				&15&16&28&22		&14&11&22&16\\    
$t_3$ 			&5&6&12&11				&11&14&22&20		&8&7&16&14\\\hline\hline
\multicolumn{1}{c}{$n=240$}		&	\multicolumn{4}{c}{}	&	\multicolumn{4}{c}{} &	\multicolumn{4}{c}{}	\\\hline
$N(0,1)$		&12&16&26&23			&32&36&49&43		&27&24&34&28\\
$L(0,1)$		&9&12&20&17				&27&30&40&32		&20&19&27&22\\ 
$NM(3,0.01)$	&10&14&22&22			&29&32&44&40		&24&22&34&26\\
$t_5$			&8&11&18&18				&26&30&40&32		&20&21&30&23\\    
$t_3$ 			&4&8&13&12				&15&22&29&22		&14&14&21&18\\\hline\hline
\multicolumn{1}{c}{$n=500$}		&	\multicolumn{4}{c}{}	&	\multicolumn{4}{c}{} &	\multicolumn{4}{c}{}	\\\hline
$N(0,1)$		&25&32&42&41			&72&73&80&74		&56&51&60&54\\
$L(0,1)$		&18&27&34&32			&53&60&65&63		&43&40&47&42\\ 
$NM(3,0.01)$	&23&33&43&38			&64&68&75&71		&50&48&56&51\\
$t_5$			&17&28&36&29			&52&61&67&64		&41&40&46&45\\    
$t_3$ 			&8&18&22&21				&30&45&48&45		&25&32&36&31\\
\hline
\end{tabular}
\end{table}

\paragraph{Analysis of power.} Tables \ref{table.power.ind.1_5}--\ref{table.power.dep.2} lists empirical rejection frequencies under various alternatives. Tables \ref{table.power.ind.1_5} and \ref{table.power.ind.2} show results for serial independence ($\rho=0$) for change-sizes $\lambda = 1.5$ and $\lambda=2$, respectively. Tables \ref{table.power.dep.1_5} and \ref{table.power.dep.2} contain figures for the serial dependence setting ($\rho=0.8$), again for change-sizes $\lambda = 1.5$ and $\lambda=2$, respectively. All tables show results for $\theta = 1/4, 1/2$, and $3/4$ and $n = 60, 120, 240$, and 500. We make the following observations:
\begin{enumerate}[(1)]
\item
All test have better power at independent sequences than dependent sequences. Note that $\rho = 0.8$ constitutes a scenario of rather strong serial dependence. 
\item
All tests loose power as the tails of the innovation distribution increase, but the loss is much more pronounced for the variance than for the other estimators.  The distributions listed in the tables are in ascending order according to their kurtosis. The kurtoses of $N(0,1)$, $NM(3,1\%)$, $L(0,1)$, $t_5$, and $t_3$ are 0, 1.63, 3, 6, and $\infty$, respectively. General formulae are given in Table \ref{tab:distributions}. 
\item
The test generally have a higher power for $\theta=3/4$ than for $\theta=1/4$. This may seem odd at first glance since in both cases the change occurs equally far away from the center of the sequence. However, since we always consider changes from a smaller ($\lambda =1$) to a larger scale ($\lambda = 1.5$ or $2$), a sequence with $\theta=1/4$ has a higher overall variability and hence yields a larger long-run variance estimate than a sequence with $\theta = 3/4$. Since we divide the test statistics by the (the root of the) long-run variance estimates, this implies a difference in power.
\item
The GMD-based test turns out to have the overall best performance, with the $Q_n^{0.8}$ and also MD not trailing far behind. In the independence setting, the $Q_n^{0.8}$ is the best for the $t_3$ distribution. 
\item
It is interesting to note that all the competitors, GMD, $Q_{0.8}$, and MD, dominate the variance-based test also under normality. The explanation is that a multiplicative change in scale, as we consider here, tends to blow up the long-run variance estimate $\hat{D}_{\sigma^2}$ much more than the corresponding estimates for the other estimators. To illustrate this, consider a sequence of i.i.d.\ $N(0,1)$-variates $X_1, \ldots, X_{\lfloor n/2 \rfloor}$ followed by a sequence of i.i.d.\ $N(0,4)$-variates $X_{\lfloor n/2 \rfloor + 1}, \ldots, X_n$. For large $n$, the quantity that $\hat{D}^2_{\sigma^2}$ estimates when applied to $X_1, \ldots, X_n$ can be seen to be $ASV(\sigma^2_n; NM(\gamma=2, \epsilon=1/2)) = E(Y^4) - E(Y^2)^2$ for $Y \sim NM(\gamma=2, \epsilon=1/2)$, whereas $\hat{D}^2_{d}$ estimates $ASV(d_n; NM(2, 1/2)) = E(Y^2) - (E|Y|)^2$ with $Y$ as before. Compared to a stationary sequence of $N(0,1)$-variates, the former quantity is blown up by the factor 19.25, the latter only by 2.93.
\end{enumerate}

%=================================================================================================================
% table power, dependent, lambda = 2
%=================================================================================================================    
\begin{table}[t]
\small
\caption{ \label{table.power.dep.2}
	  {\it Test power.} Rejection frequencies (\%) at asymptotic 5\% level.  
	Same set-up as Table~\ref{table.power.dep.1_5} except scale changes by {\bf factor $\lambda = 2.0$}. 
}

\smallskip
\centering
\begin{tabular}{c|r@{\,\,\,}r@{\,\,\,}r@{\,\,\,}r|r@{\,\,\,}r@{\,\,\,}r@{\,\,\,}r|r@{\,\,\,}r@{\,\,\,}r@{\,\,\,}r} 
\hline
Change location: & \multicolumn{4}{|c}{$[n/4]$}& \multicolumn{4}{|c}{$[n/2]$} & \multicolumn{4}{|c}{$[n3/4]$}\\\hline
 Estimator:	&  Var & MD & GMD & $Q_n^{0.8}$ &  Var & MD & GMD & $Q_n^{0.8}$ &  Var & MD & GMD & $Q_n^{0.8}$ \\\hline\hline
\multicolumn{1}{c}{ $n=60$}		&	\multicolumn{4}{c}{}	&	\multicolumn{4}{c}{} &	\multicolumn{4}{c}{}	\\\hline
$N(0,1)$		&13&11&27&26			&24&18&41&39		&16&10&27&22\\
$L(0,1)$		&14&12&26&21			&20&15&38&34		&16&8&25&20\\ 
$NM(3,0.01)$	&12&11&27&24			&22&16&41&39		&18&10&29&20\\
$t_5$			&10&10&22&25			&18&14&36&34		&15&8&25&18\\   
$t_3$ 			&8&8&18&21				&15&12&31&29		&15&9&23&21\\\hline\hline
\multicolumn{1}{c}{$n=120$}		&	\multicolumn{4}{c}{}	&	\multicolumn{4}{c}{} &	\multicolumn{4}{c}{}	\\\hline
$N(0,1)$		&12&16&30&29			&36&38&59&51		&28&25&41&26\\
$L(0,1)$		&10&11&24&23			&26&31&50&38		&25&21&38&24\\ 
$NM(3,0.01)$	&11&13&27&26			&33&38&56&47		&27&22&40&24\\
$t_5$			&10&12&25&24			&26&29&48&41		&24&20&37&24\\     
$t_3$ 			&7&8&18&19				&20&24&40&33		&15&14&28&18\\\hline\hline
\multicolumn{1}{c}{$n=240$	}	&	\multicolumn{4}{c}{}	&	\multicolumn{4}{c}{} &	\multicolumn{4}{c}{}	\\\hline
$N(0,1)$		&20&32&52&44			&70&77&87&80		&63&60&74&54\\
$L(0,1)$		&14&22&37&37			&58&70&79&66		&52&51&64&42\\ 
$NM(3,0.01)$	&16&29&45&43			&67&78&85&74		&57&55&68&51\\
$t_5$			&14&24&36&35			&54&68&78&67		&48&49&60&44\\     
$t_3$ 			&8&17&26&23				&33&49&59&54		&35&37&47&30\\\hline\hline
\multicolumn{1}{c}{$n=500$}		&	\multicolumn{4}{c}{}	&	\multicolumn{4}{c}{} &	\multicolumn{4}{c}{}	\\\hline
$N(0,1)$		&56&80&88&85			&98&99&100&99		&95&94&97&96\\
$L(0,1)$		&37&64&74&70			&92&97&97&96		&88&89&93&87\\ 
$NM(3,0.01)$	&45&76&83&77			&95&99&99&99		&91&91&94&93\\
$t_5$			&34&65&74&73			&88&96&97&98		&84&87&90&86\\     
$t_3$ 			&18&45&51&50			&64&88&88&87		&60&72&76&67\\
\hline
\end{tabular}
\end{table}

%\FloatBarrier %zwingt die Grafiken in den obigen Kapitel zu bleiben

\section{Data Example}
\label{section_data_example}

We consider two data examples: a hydrological and a financial time series. According to our knowledge, neither has been analyzed in a change-point context. 
\begin{figure}[t]
%\begin{minipage}{\textwidth}
\includegraphics[width=\textwidth]{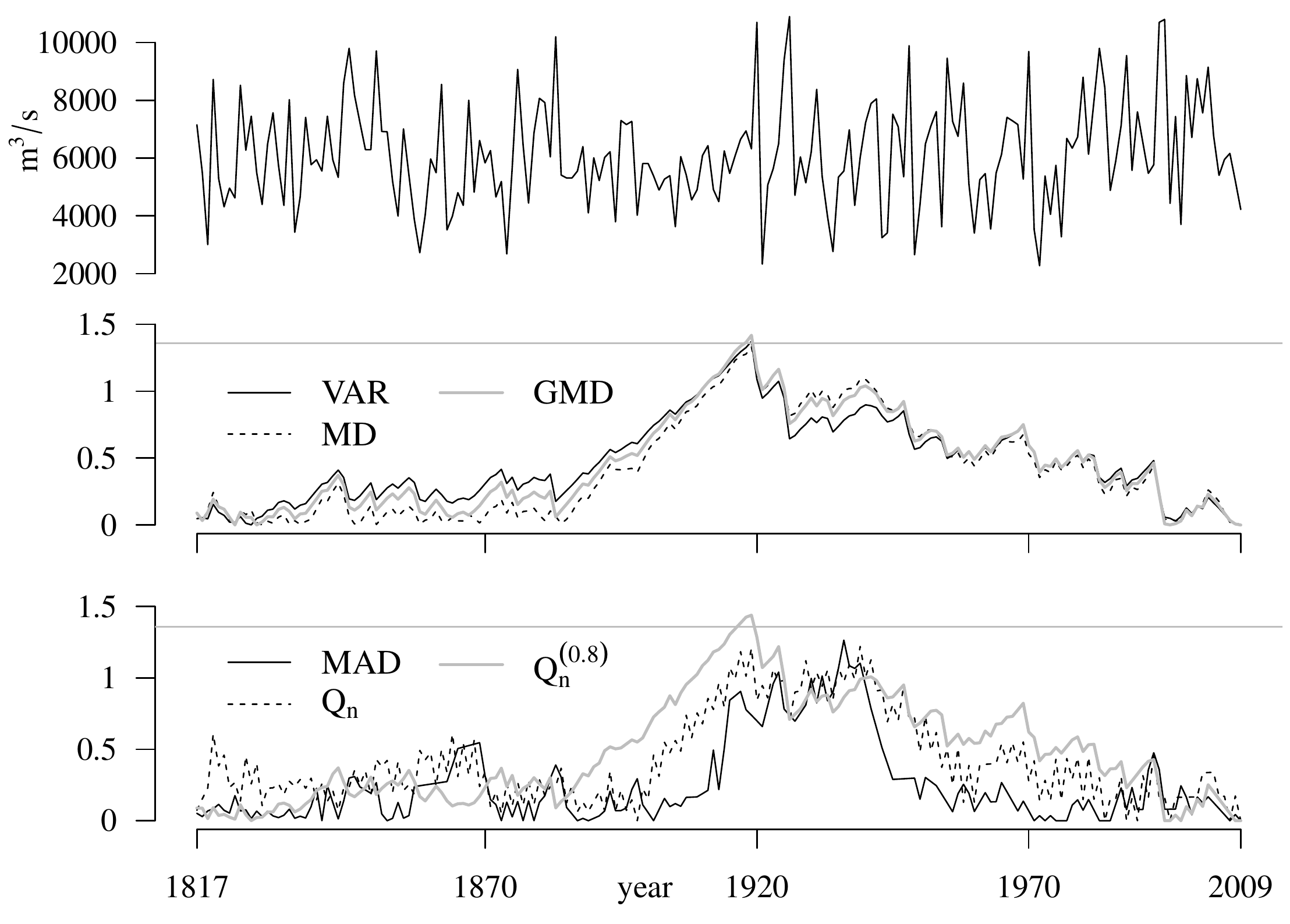}
\caption{ \label{data_tests_example_hydrology}
	Top row: annual maximum discharge ($m^3/s$) of the river Rhine at Cologne, Germany, between 1817 and 2009.
	Middle and bottom rows: change-point processes 
	 $\hat{D}^{-1}_s \left( k/\sqrt{n} |s_{1:k}-s_{1:n}|\right)_{2\le k \le n}$
	for estimators $s_n = \sigma^2_n$, $d_n$, $g_n$ (middle) and $m_n$, $Q_n$, and $Q_n^{0.8}$ (bottom); HAC kernel bandwidth: $b_n = 4$.
}
\end{figure}
\begin{figure}
\includegraphics[width=\textwidth]{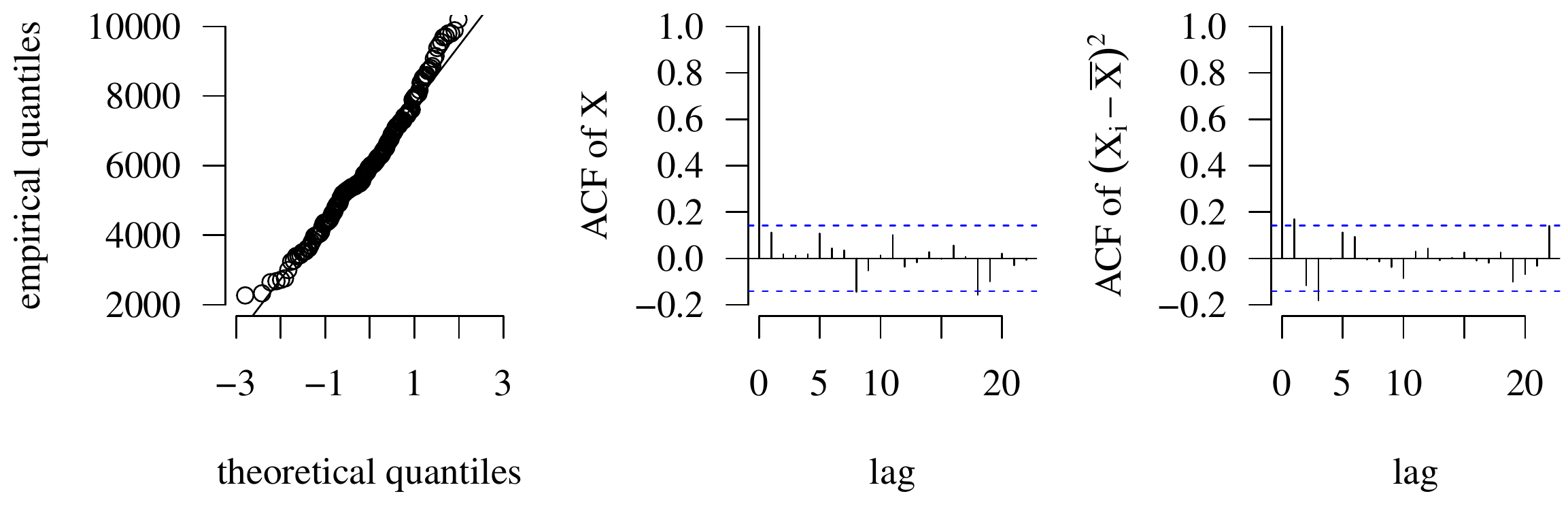}
\caption{\label{qq_plot_example_hydrology}
Hydrology data from Figure \ref{data_tests_example_hydrology}: Normal q-q plot, ACF of data, and ACF of squared centered data.
}
%\end{minipage}
\end{figure}
The first data set consists of the annual maximum discharge (in cubic meters per second) of the river Rhine at Cologne, Germany, in the years 1817 to 2009 ($n=193$). The time series is plotted in Figure \ref{data_tests_example_hydrology}, top row. Figure \ref{qq_plot_example_hydrology} depicts a normal q-q plot, which reveals that the marginal distribution is fairly normal. 
%agreement with a normal distribution. 
Furthermore, the autocorrelation function of the data and the squared mean-centered data are plotted. They indicate a weak serial dependence. We thus apply the change-point tests from the previous section with long-run variance estimation settings as before except we set the HAC-bandwidth to $b_n = 4$. The latter is in consistency with the other data example. The results are very similar for smaller bandwidths.

The change-point test processes $\hat{D}^{-1}_s \left( k/\sqrt{n} |s_{1:k}-s_{1:n}| \right)_{2\le k \le n}$ show a fair agreement for $s_n = \sigma^2_n, d_n, g_n$, and $Q_n^{0.8}$, cf.\ Figure \ref{data_tests_example_hydrology}, middle and bottom row. All attain their maxima at 1919 with p-values ranging from 0.021 to 0.046, i.e., they confirm the existence of a change in scale around 1920, which is suggested by a visual inspection of the series.
This change coincides with the implementation a variety of structural river works upstream from Cologne, particularly along the Rhine's tributaries Main and Neckar in the early 1920s.
For illustration purposes, we also plot the corresponding curves for the original $Q_n$ ($\alpha = 1/4$) and the MAD in the bottom row of Figure \ref{data_tests_example_hydrology}. Both are very rugged, distinctively different from the other curves, and yield p-values above 5\%.

\begin{figure}[t]
\includegraphics[width=\textwidth]{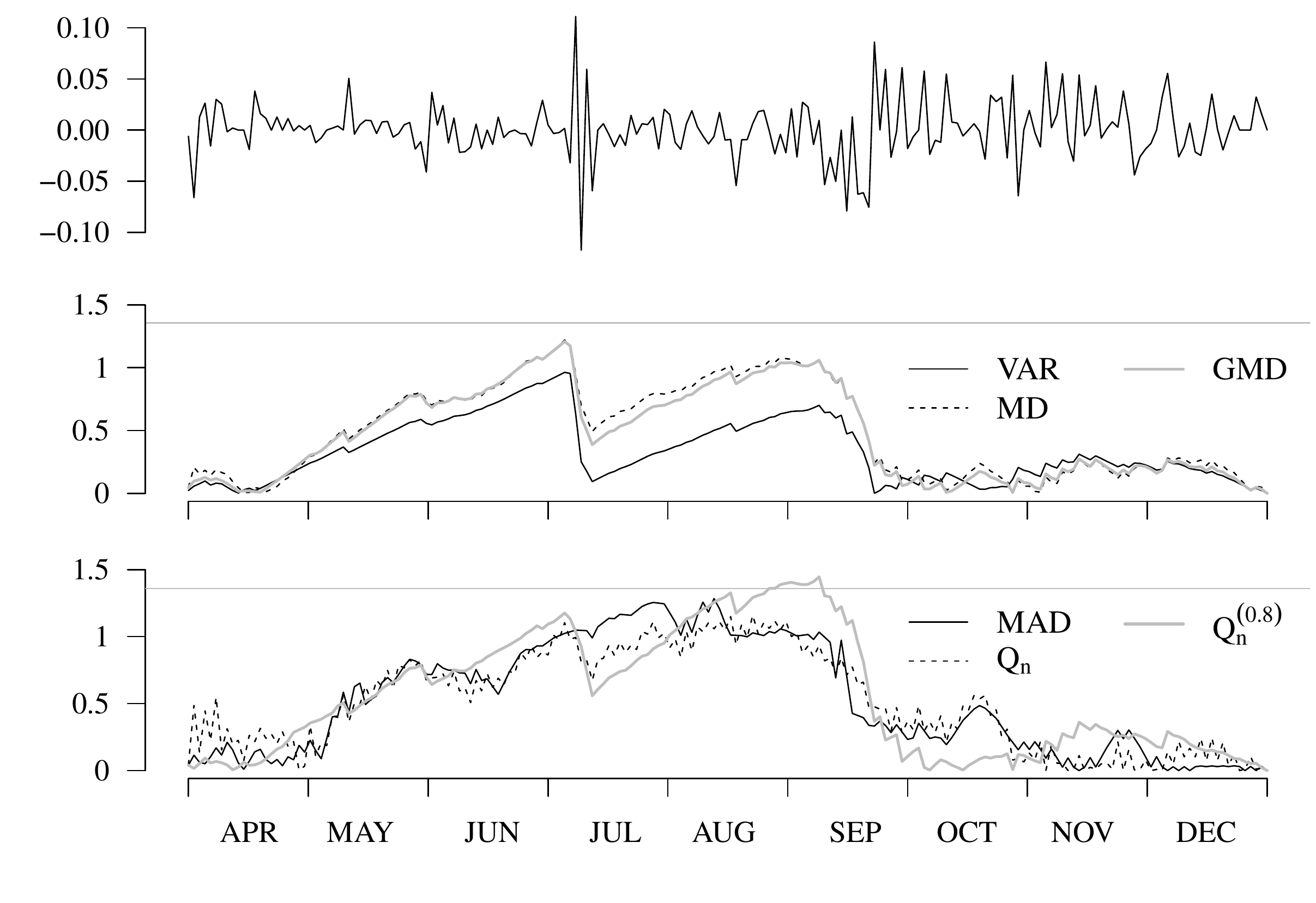}
\caption{ \label{data_tests_example_VW}
	Top row: daily log returns of VW share from April to December 2001.
	Middle and bottom rows: change-point processes 
	 $\hat{D}^{-1}_s \left( k/\sqrt{n} |s_{1:k}-s_{1:n}|\right)_{2\le k \le n}$
	for estimators $s_n = \sigma^2_n$, $d_n$, $g_n$ (middle) and $m_n$, $Q_n$, and $Q_n^{0.8}$ (bottom); HAC kernel bandwidth: $b_n = 4$.
} 
\end{figure}

\begin{figure}[ht]
\includegraphics[width=\textwidth]{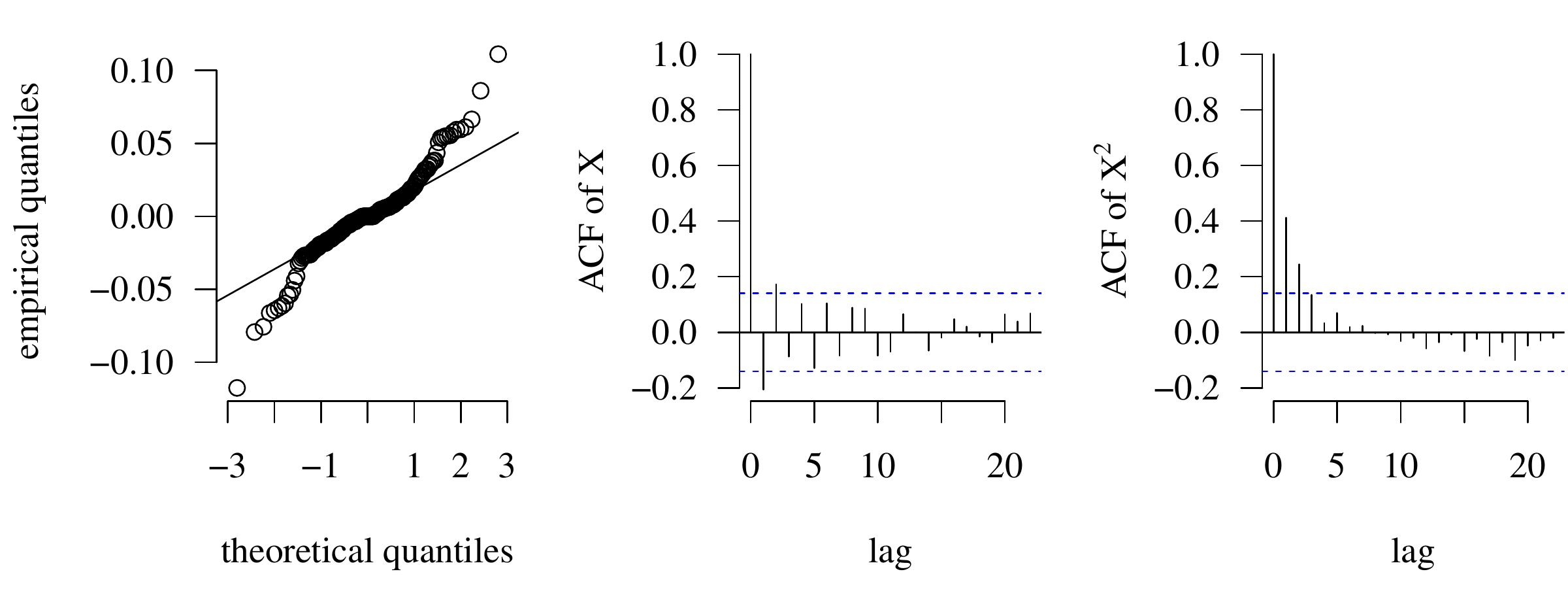}
\caption{ \label{qq_plot_example_VW}
 Financial data from Figure \ref{data_tests_example_VW}: Normal q-q plot, ACF of data, and ACF of squared data.
}
\end{figure}

The second example consists of the log returns of the daily closings of the Volkswagen share, traded at the German stock exchange in Frankfurt, within the last three quarters of the year 2001 ($n=196$). The impact of 9/11 on the volatility of the series is clearly visible, cf.\ Figure \ref{data_tests_example_VW}, top row. 
In the ``normal'' situation of the first example, all tests  were in fair agreement (except the previously discarded MAD and $Q_n$). This example differs in a variety of features, which shall illustrate the differences of the tests. The normal q-q plot shows heavier than normal tails, and the auto-correlation function of the squared sequence reveals some serial dependence (cf.\ Figure \ref{qq_plot_example_VW}). Furthermore, there is a short period of strong oscillation from July 10--12, 2001 (the reason for which is not known to us). 
Removing these three dates from the series, all tests consistently detect a change with p-values of 1\% and less and a place it at the beginning of September. However, \emph{with} those three days, the variance-based change-point test process $\hat{D}^{-1}_{\sigma^2} \left( k/\sqrt{n} |\sigma^2_{1:k}-\sigma^2_{1:n}| \right)_{2\le k \le n}$ attains its maximum at July 5 and yields a p-value of 0.31.
This is an example where few outliers mask an apparent change. The tests based on MD and GMD behave in principle similarly, but provide some mild evidence for a change with p-values of about 0.10. The $Q_n^{0.8}$, however, gives a p-value of 0.03, and the maximum is attained at September 7.

The curves for the MAD and $Q_n$ are plotted as well. Both curves, again, look distinctively different from the others, and the respective tests do not provide strong evidence for a change. The tests were carried out, as before, with $b_n = 4$ and the other long-run variance estimation parameters as in Section \ref{section_simulation_results}. All simulations and data analyses were performed in R \citep{R}. %, using the package robustbase \citep{robustbase} for the original $Q_n$.

%\FloatBarrier %zwingt die Grafiken in den obigen Kapitel zu bleiben

%/////////////////////////////////////////////////////////////////////////////////////////////////////////////////////
%/////////////////////
%////////////
%/////////
%///////
%\\\\\\\
%\\\\\\\\\
%\\\\\\\\\\\\
%\\\\\\\\\\\\\\\\\\\\\
%\\\\\\\\\\\\\\\\\\\\\\\\\\\\\\\\\\\\\\\\\\\\\\\\\\\\\\\\\\\\\\\\\\\\\\\\\\\\\\\\\\\\\\\\\\\\\\\\\\\\\\\\\\\\\\\\\\\\

\section{Conclusion and Outlook}\label{section_conclusion_and_outlook}

We have studied the problem of detecting changes in scale beyond the established sum-of-squares methodology. We have considered test statistics based on alternative scale measures, which have a better outlier-resistance and and a better efficiency at heavy-tailed distributions than the sample variance: the mean deviation, the median absolute deviation (MAD), Gini's mean difference, and the $Q_n^\alpha$. The MAD and the original $Q_n$ (i.e., $Q_n^\alpha$ with $\alpha = 1/4$) may be confidently discarded for the purpose of change-point detecting, whereas the mean deviation, Gini's mean difference, and the $Q_n^\alpha$ for $0.7 \le \alpha \le 0.9$ provide very good alternatives, that improve upon the classical test not only under heavy tails but also under normality. We found Gini's mean difference and the $Q_n^\alpha$, which are both based on pairwise differences, to altogether outperform the mean deviation. Our general recommendation is to use the Gini's mean difference based test in case of normal or near-normal data situations, and to use the $Q_n^{0.8}$ or $Q_n^{0.75}$ if the occurrence of gross errors is suspected.

We have confined our investigation here to a number of popular and explicit scale estimators. The choice of the estimators is
based upon a prior assessment of the efficiency and robustness (where we understand the latter as ``high efficiency over a broad range of distributions'' rather than ``high breakdown-point'') of various scale estimators, and we suspect that the tests may not be substantially improved by plugging in other scale estimators. Nevertheless the question arises, what other alternative are possible and might be used for the problem at hand. 
For an overview of approaches to robust scale estimation we refer the reader to \citet[][Chapter 5]{Huber_2009}. For a numerical comparison of scale estimators, see, e.g., \citet{Lax1985}. % or \citet{Gerstenberger_Vogel}.
One common way to safeguard against outliers is to use truncation or Winsorization. \citet{Lee_Park_2001} consider a version of the $\sigma^2$-based test based with truncated observations. Their simulation results suggest quite a substantial loss in power under normality, whereas we observe the opposite for our robustification approach. 
There is another conceptual advantage of pairwise differences: there is no location to estimate. At skewed distributions, taking, e.g., the mean or the median leads to distinctively different scale estimators. This ambiguity does not exist for pairwise-difference-based estimators.

Gini's mean difference and the $Q_n^\alpha$ are based on 
 the kernel $g(x,y) = |x-y|$ of order two. We have noted that both have a considerably higher efficiency under normality as compared to the respective estimators based on distances to the central location. So, may kernels of higher order even be better? A general observation seems to be that this is not the case, see, e.g., \citet[][Section 4]{Rousseeuw_Croux_1992}. Higher-order-kernels require a higher computational effort, but tend to provide neither better efficiency nor robustness.

A crucial part of all change-point tests for dependent data is the estimation of the long-run variance. We have proposed estimators  based on HAC kernel estimation, which is common in the change-point context. An alternative estimation technique is block sub-sampling, see, e.g., \citet{dehling:fried:sharipov:vogel:wornowizki:2013} and the references therein. An entirely different approach, which avoids any unknown scaling constants in the limit distribution of the test statistic is the self-normalization approach as proposed by \citet{shao:zhang:2010}.

Yet an alternative way of assessing critical values is to estimate the distribution of the test statistic using bootstrapping. A variety of bootstrap procedures have been proposed for dependent data, e.g., the block bootstrap \citep{kuensch:1989}, the stationary bootstrap \citep{politis:romano:1994}, the tapered block bootstrap \citep{Paparoditis.2001},  or the dependent wild (or multiplier) bootstrap \citep{shao:2010}.

The block bootstrap for U-statistics (such as Gini's mean difference) has been studied by \citet{dehling2010central}. Recently, \citet{buecher:kojadinovic:2016} showed the consistency of the dependent multiplier bootstrap for U-statistics and also established the validity of dependent multiplier bootstrap procedures for change-point test statistics based on this class of statistics. For quantiles, the block bootstrap was investigated by \citet{sun2006bootstrapping} and the multiplier bootstrap by \citet{doukhan2015dependent}. For U-quantiles, such as the $Q_n^\alpha$, we are unaware of any work concerning bootstrap methods. 

Finally, the consideration of high-breakdown-point estimators such as the MAD and the original $Q_n$, which have turned out to be rather unsuited for the problem at hand, prompts the question of what kind of robustness can be expected and is desired in a change-point context, and if this can be mathematically formalized and quantified. 

%/////////////////////////////////////////////////////////////////////////////////////////////////////////////////////
%/////////////////////
%////////////
%/////////
%///////
%\\\\\\\
%\\\\\\\\\
%\\\\\\\\\\\\
%\\\\\\\\\\\\\\\\\\\\\
%\\\\\\\\\\\\\\\\\\\\\\\\\\\\\\\\\\\\\\\\\\\\\\\\\\\\\\\\\\\\\\\\\\\\\\\\\\\\\\\\\\\\\\\\\\\\\\\\\\\\\\\\\\\\\\\\\\\\

\section*{Acknowledgement}
The researcher were supported by the Collaborative Research Centre 823 \emph{Statistical modelling of nonlinear dynamic processes} and the Konrad-Adenauer-Stiftung. The authors thank Svenja Fischer for providing the river Rhine discharge data set and Marco Thiel for the stock exchange data set.

%/////////////////////////////////////////////////////////////////////////////////////////////////////////////////////
%/////////////////////
%////////////
%/////////
%///////
%\\\\\\\
%\\\\\\\\\
%\\\\\\\\\\\\
%\\\\\\\\\\\\\\\\\\\\\
%\\\\\\\\\\\\\\\\\\\\\\\\\\\\\\\\\\\\\\\\\\\\\\\\\\\\\\\\\\\\\\\\\\\\\\\\\\\\\\\\\\\\\\\\\\\\\\\\\\\\\\\\\\\\\\\\\\\\

\appendix
%\begin{appendix}

%\color{blue}
\section{Supplementary material for Section \ref{section_qn}} 
\label{app_supp}

\subsection{The population value and the asymptotic variance of $Q_n^\alpha$}
\label{app_supp_1}

\paragraph{General expressions.} Let $X,Y \sim F$ be independent and let $F$ possess a Lebesgue density $f$. Generally, $Q^\alpha = Q^\alpha(F) = U^{-1}(\alpha)$, where $U$ is the distribution function of $|X-Y|$ and $U^{-1}$ the corresponding quantile function, i.e., $U^{-1}(\alpha) = \inf \{ t \,|\, \alpha \le U(t) \}$.

For the asymptotic variance of the $Q_n^\alpha$, we analyze (\ref{eq:lrv.qn}) for independent sequences $X_1,\ldots, X_n \sim F$:
\[
	ASV(Q_n^\alpha; F) = \frac{4 \, E \psi^2(X)}{u^2(Q^\alpha)}, 
	\qquad X \sim F,
\]
where
\[
	E\psi^2(X) = \int_R \left\{ F(x+Q^\alpha)- F(x-Q^\alpha) \right\}^2 f(x) d x - \alpha^2,
\]
and $u$ is the density of $|X-Y|$, i.e., 
\be \label{eq:u}
	u(x) = 2 \int_\R f(t)f(t-x) d t, \qquad x \ge 0. 
\ee
The corresponding cdf $U$ is then given by 
\be \label{eq:U}
	U(x) = \int_0^x u(y) dy, , \qquad x \ge 0. 
\ee

\paragraph{Specific expressions.} For distributions $F$, where the convolution with itself admits an tractable form, more explicit expressions for the above terms can be given. For instance, we have for the standard normal distribution $N(0,1)$ withd density $\phi$ and cdf $\Phi$:
\[
	Q^\alpha = \sqrt{2} \Phi^{-1}\left\{(\alpha+1)/2\right\}, \quad u(x) = \sqrt{2} \phi(x/\sqrt{2}), \ (x \ge 0).
\]
For the $t_1$ or Cauchy distribution we obtain
\[
 Q^\alpha = 2 \tan(\pi \alpha / 2), \quad u(x) = 4\pi^{-1}(x^2+4)^{-1}, \ (x \ge 0), 
\]
and 
\[
	ASV(Q_n^\alpha; t_1) = \frac{\pi}{4}\{(Q^\alpha)^2 +4\}^2 \left\{ \int_\R \frac{ \pi^2 [\arctan(x+Q^\alpha)-\arctan(x-Q^\alpha)]^2}{1+\x^2}  d x - \alpha^2  \right\}. 
\]
For the Laplace distribution $L(0,1)$, cf.\ Table~\ref{tab:distributions}, we get
\[
	u(x) = 1/2 (1+x) e^{-x}, (x \ge 0), \quad U(x) = 1 - e^{-x} - x/2 e^{-x}, (x\ge 0), 
\]
$Q^\alpha$ is obtained by numerically solving $U(x)= \alpha$, and 
\[
	ASV(Q_n^\alpha;L(0,1)) 
	= 16 \,
	\frac{1- (1+Q^\alpha) e^{-Q^\alpha} - 1/6\, e^{-Q^\alpha}\{ 1-e^{-Q^\alpha}\}^2 - \alpha^2}
	      { \left(1+Q^\alpha\right)^2 e^{- 2 Q^\alpha}}.
\]

\subsection{The asymptotic variance of the $t_\nu$ maximum-likelihood estimator for scale.} 

In Figure \ref{plot_qn_are_alpha}, the asymptotic relative efficiencies of the $Q_n^\alpha$ at various distributions with respect to the respective scale maximum likelihood estimators are displayed as a function of $\alpha$. The asymptotic variance of the normal scale MLE (the standard deviation) at $N(0,1)$  is 1/2, the asymptotic variance of the Laplace scale MLE (the mean deviation) at $L(0,1)$ is 1. Lemma \ref{lem:1} below yields the asymptotic variance of the $t_\nu$ scale MLE. The authors believe that this result is known in the literature, but were not able to find a reference for exactly this result and would be grateful for any information regarding an earlier reference.

Consider, for fixed $\nu > 0$, the $t_\nu$ scale family, i.e., the parametric family of densities $\{ f_{\nu,s} \,| \, s > 0 \}$ with 
\[
	f_{\nu,s}(x) = \frac{c_\nu}{s}\left( 1+ \frac{x^2}{\nu s^2}\right)^{-(\nu+1)/2}, \quad x \in \R, \qquad \mbox{ with } \ c_\nu = \sqrt{\nu}B(\nu/2,1/2),
\]
where $B$ denotes the beta function. Let $F_{\nu,s}$ denote the corresponding distribution function. Then, the maximum likelihood estimator $\hat{s}_n$ for $s$ at an i.i.d.\ sample $X_1, \ldots, X_n \sim F_{\nu,s}$ satisfies $\sqrt{n}(\hat{s}_n - s) \cid N(0, I(\nu,s)^{-1})$, where 
%\[
	 $I(\nu,s) = E\left[\left\{\partial/(\partial s) \log f_{\nu,s}(X)\right\}^2 \right]$
%\]
with $X \sim F_{\nu,s}$ denotes the Fisher information about $s$ contained in $X$.

\begin{lemma} $I(\nu,s) = 2\nu(\nu+3)^{-1}s^{-2}$. \label{lem:1}
\end{lemma}
\begin{proof}
$I(\nu,s) = s^{-2} I(\nu,1)$
\[
	= \frac{c_\nu}{s^2} \int_\R \left(\frac{(\nu+1) x^2}{\nu+x^2} -1 \right)\left(1 + \frac{x^2}{\nu}\right)^{-\frac{\nu+1}{2}}  d x
	= \frac{c_\nu}{s^2} \int_\R (x^2-1)^2 \left(1 + \frac{x^2}{\nu}\right)^{-\frac{\nu+5}{2}} d x
\]
Substituting $y = \sqrt{(\nu+4)/\nu}\,x$, we transform this integral to 
\[
	\frac{c_\nu}{s^2} \sqrt{\frac{\nu}{\nu+4}} \int_\R 
	\left\{ \left(\frac{\nu}{\nu+4}\right)^2 y^4 - 2 \frac{\nu}{\nu+4} y^2  + 1 \right\} \frac{1}{c_{\nu+4}} f_{\nu+4,1}(y) d y
\]
For $Y \sim t_d$, we have $EY^4 = 3 d^2 (d-2)^{-1} (d-4)^{-1}$, $d > 4$, and $EY^2 = d/(d-2)$, $d > 2$. Applying this with $d = \nu + 4$, we obtain
\[
	s^{-2} I(\nu,1) = \frac{1}{s^2} \frac{c_\nu}{c_{\nu+4}} \sqrt{\frac{\nu}{\nu+1}}\, \frac{2(\nu+1)}{\nu+2}.
\]
With $c_\nu = \Gamma (\frac{\nu+1}{2})/(\sqrt{\nu \pi}\,\Gamma( \frac{\nu}{2}))$ we arrive at the expression given in Lemma \ref{lem:1}.
\end{proof}
Remark: If we consider (more realistically) the location-scale $t_\nu$-family, where location and scale are estimated simultaneously, the scale maximum-likelihood estimator $\hat{s}_n$ has the same asymptotic variance $(\nu+3)s^2/(2\nu)$ as computed above. Due to the symmetry of the distribution, the maximum-likelihood estimators of location and scale are uncorrelated by an invariance argument and the $2\times 2$ Fisher information matrix is diagonal. For a discussion on the uniqueness of the solutions to the $t_\nu$ likelihood equations in the location-and-scale case and the scale-only scale in the univariate as well as the multivariate setting, see \citet{kent:tyler:1991} and the references therein.

%/////////////////////////////////////////////////////////////////////////////////////////////////////////////////////
%/////////////////////
%////////////
%/////////
%///////
%\\\\\\\
%\\\\\\\\\
%\\\\\\\\\\\\
%\\\\\\\\\\\\\\\\\\\\\
%\\\\\\\\\\\\\\\\\\\\\\\\\\\\\\\\\\\\\\\\\\\\\\\\\\\\\\\\\\\\\\\\\\\\\\\\\\\\\\\\\\\\\\\\\\\\\\\\\\\\\\\\\\\\\\\\\\\\

\section{Proofs}
\label{app_proofs}

Before stating the proofs of Theorems \ref{th:gmd} and \ref{th:qn}, we want to provide an intuitive explanation for the expressions (\ref{eq:lrv.gmd}) and (\ref{eq:lrv.qn}) of the long-run variances for Gini's mean difference and the $Q_n^\alpha$, respectively. 
Let $(X_i)_{i\in\Z}$ be a stationary sequence with marginal distribution $F$, and furthermore $X, Y \sim F$ be independent. 
Gini's mean difference is a U-statistic with kernel $g(x,y) = |x-y|$. The main tool for deriving asymptotics for U-statistics is the Hoeffding decomposition. By letting
\[
	g_0 = E g(X,Y), \quad g_1(x) = E g(x,Y) - g_0, \quad g_2(x,y) = g(x,y) - g_1(x) - g_1(y) - g_0, 
\]
we can decompose the U-statistic 
\[
  G_n = \binom{n}{2}^{-1} \sum_{1\le i < j \le n} g(X_i, X_j)
\]
into 
\[
 \sqrt{n}(G_n - g_0) = 
 \frac{2}{\sqrt{n}}\sum_{i=1}^{n}g_1(X_i) + 
 \frac{2}{\sqrt{n}(n-1)}\sum_{1\leq i<j\leq n}g_2(X_i,X_j).
\]
The first term on the right-hand side is called the linear part and the second term is the degenerate part of the U-statistic. Under appropriate regularity conditions, the degenerate part can be seen to vanish asymptotically, and hence the asymptotic variance of $G_n$ is determined by that of $\frac{1}{\sqrt{n}}\sum_{i=1}^{n}g_1(X_i)$, which, by a central limit theorem for dependent sequences, can be seen to be 
$\sum_{i=-\infty}^{\infty} E \left\{g_1(X_0)g_1(X_h)\right\}$. The fact that $\frac{1}{\sqrt{n}}\sum_{i=1}^{n}g_1(X_i)$ appears twice in the Hoeffding decomposition explains the factor 4 in (\ref{eq:lrv.gmd}) and (\ref{eq:lrv.qn}).  The asymptotic variance of $G_n$ carries over to the ``change-point process'' $\left\{ k n^{-1/2}(G_k - G_n)\right\}_{2\le k \le n}$. A detailed proof for general U-statistics under the dependence scenario considered here is given by \citet{DVWW14}.

The estimator $Q_n^\alpha$ is a U-quantile, i.e., instead of taking the first sample moment of $|X_i-X_j|$, $1 \le i < j \le n$, we consider the sample $\alpha$-quantile of these $\binom{n}{2}$ values. An essential step in the asymptotic analysis of U-quantiles %(just as of sample quantiles in general)
 is to relate the empirical quantile function $U_n^{-1}$ of $|X_i-X_j|$, $1 \le i < j \le n$, to the corresponding empirical distribution function $U_n$ by means of a generalized Bahadur representation:
\[
	Q_n^\alpha = U_n^{-1}(\alpha) = Q^\alpha  + \frac{\alpha - U_n(Q^\alpha)}{u(Q^\alpha)} + R_n, 
\]
where $Q^\alpha$ and $u$ are as in Section \ref{app_supp_1}, and $R_n$ is remainder term, which, under appropriate regularity conditions, converges to zero sufficiently fast. Then, recalling that
\[
	U_n(t) = \binom{n}{2}^{-1} \sum_{1\le i < j \le n} \Ind{|X_i-X_j|\le t}, \qquad t \in \R, 
\]
yields
\[
	\sqrt{n} (Q_n^\alpha - Q_\alpha) =  \frac{\sqrt{n}}{u(Q^\alpha)} 
	\, \binom{n}{2}^{-1}\!\! \sum_{1\le i < j \le n} \left(\alpha  - \Ind{|X_i-X_j|\le Q^\alpha} \right)
	\, + \, \sqrt{n} R_n, 
\]
This explains on the one hand the appearance of the density $u(x)$ in the long-run variance (\ref{eq:lrv.qn}) %of the $Q_n$. %The second term vanishes asymptotically. 
%Thus the Bahadur representation 
and on the other hand relates the U-quantile $Q^\alpha$ to a U-statistic with kernel $h(x,y) = \Ind{|X_i-X_j|\le Q^\alpha}$, which is then further treated by the Hoeffding decomposition as described above. The function $\psi(x)$ in (\ref{eq:lrv.qn}) can be seen to be the linear kernel $g_1(x)$ associated with the kernel $h(x,y)$. 

\begin{proof}[Proof of Theorem \ref{th:gmd}] The theorem is a corollary of Corollary 2.8 of \citet{DVWW14}. There are four conditions imposed there, labeled Assumptions 2.2, 2.3, 2.4 and 2.6. Assumption 2.2 and 2.3 in \citet{DVWW14} are essentially our Assumption \ref{ass:gmd} (a) and (b). \citet[][Assumption 2.3]{DVWW14} is a condition on the moments on the kernel $g(X_i,X_j)$, which in our case $g(x,y) = |x-y|$ directly translate into moment condition on the data itself. Assumption 2.6 concerns the long-run variance and is identical to our Assumption 3.1. Thus it remains 
\citet[][Assumption 2.4]{DVWW14}, which is as follows:
\begin{assumption} \label{ass:variation}
There are positive constants $L,\epsilon_0>0$ such that for all $\epsilon\in (0,\epsilon_0)$
\[ %begin{equation}\label{variation_condition_g}
E\left(\sup_{|x-X|\leq \epsilon, |y-Y|\leq \epsilon}|g(x,y)-g(X,Y)|\right)^2\leq L\epsilon,
\] %end{equation}
where $X,Y$ are independent with the same distribution as $X_0$.
\end{assumption}
This conditions is also known as the variation condition. It can be understood as a form of Lipschitz continuity of the kernel $g$ with respect to $F$. Since in our case $g(x,y) = |x-y|$, $g$ itself is Lipschitz continuous, this condition is fulfilled for any distribution $F$.
\end{proof}

\begin{proof}[Proof of Theorem \ref{th:qn}]
This is a corollary of Corollary 2.5 (B) of \citet{VW14}. This statement required six conditions, which are labeled Assumptions 1 to 6. Assumptions 1, 4, and 5 in \citet{VW14} coincide with our Assumptions 3.4, 3.2, and 3.1, respectively. They concern the serial dependence, the kernel density estimation, and the HAC kernel estimation, respectively. Furthermore, \citet[][Assumption 6]{VW14} is Assumption \ref{ass:variation} above, which is fulfilled since the kernel $g(x,y)=|x-y|$ is Lipschitz continuous. Assumption 2 in \citet{VW14} is a similar variation condition for the kernel $h_t(x,y) = \Ind{|x-y|\le t}$, which is required to hold not only for $t = Q^\alpha$, but, slightly stronger, for all $t$ in some neighborhood of $Q^\alpha$. The boundedness of $f$ (our Assumption \ref{ass:smoothness.F}(a)) is sufficient for this. Thus it remains to show Assumption 3 of \citet{VW14}, which is the following smoothness condition on the distribution function $U(t) = P(|X-Y|\le t)$, $t \ge 0$, with $X, Y \sim F$ independent. 
\begin{assumption} \label{ass:smoothness.U}
Let $U$ be differentiable in a neighborhood of $Q_\alpha = U^{-1}(\alpha)$ with $u(t) = U'(t)$. Furthermore,  
\begin{enumerate}[(a)]
\item there are constants $c, C > 0$ such that $c \le u(t) \le C$ in a neighborhood of $Q_\alpha$ and
\item $\displaystyle \left\vert \frac{U(t)-\alpha}{t-Q_\alpha} - u(Q_\alpha)   \right\vert = O\left(|t-Q_\alpha|^{1/2}\right)$ \ for \ $t \to Q_\alpha$.
\een
\end{assumption}
It remains to show that this condition is implied by Assumption \ref{ass:smoothness.F}. 
Recall the representations (\ref{eq:U}) and (\ref{eq:u}) for $U$ and the corresponding density $u$, respectively.
Assumptions \ref{ass:smoothness.F} (a) and (b), i.e., $f$ is bounded and its support ``has no gaps'', ensure that $u(t)$ stays away from 0 and $\infty$ at any point strictly between $\inf\{x | u(x) > 0 \}$ and $\sup \{ x| u(x) > 0 \}$, i.e., Assumption \ref{ass:smoothness.U} (a). 
From Assumptions \ref{ass:smoothness.F} (a) and (c), we find that 
%\[
	$|u(x) - u(y)| \le 2 m ||f||^2_\infty |x - y|$,
%\]
where $m$ denotes the number of intervals. Hence $u$ is Lipschitz continuous. Furthermore, for $x,y$ close enough such that $u$ is monotonic between both (and without loss of generality we assume $u(x) \le u(y)$), we have 
%\[
	$u(x) \le \left\vert (U(y)-U(x))/(y-x)\right\vert \le u(y)$
%\]
and hence 
\[
	\left\vert \frac{U(y)-U(x)}{y-x} - u(x) \right\vert = O\left(|u(x)-u(y)|\right) = O(|x-y|) \quad \mbox{ as } x \to y
\]
and hence \ref{ass:smoothness.U}(b) holds. This completes the proof.
\end{proof}

\color{black}

%\end{appendix}
\bibliographystyle{abbrvnat}
%\nocite{*}
\footnotesize
%\bibliography{Literaturverzeichnis}

\end{document}